\documentclass[journal,10pt]{article}
\usepackage[a4paper,left=0.9in,right=0.9in,top=.9in,bottom=.9in]{geometry}
\usepackage{booktabs}
\usepackage[export]{adjustbox}
\usepackage{amsmath,amssymb,amsmath}
\usepackage{algorithm,algorithmic}
\usepackage{multirow}
\usepackage{siunitx}
\usepackage{blindtext}
\usepackage[utf8]{inputenc}
\usepackage{commath}
\usepackage{graphicx}
\usepackage{xcolor}
\usepackage{nicefrac} 
\usepackage{url}
\usepackage{textpos}
\usepackage{amsthm}
\usepackage{multicol}
\usepackage[colorlinks = true,
linkcolor = black,
urlcolor  = black,
citecolor = black,
anchorcolor = black]{hyperref}
\usepackage[font={small,it}]{caption}
\usepackage{lineno}
\usepackage{textpos}

\usepackage[T1]{fontenc}
\usepackage{xcolor}

\usepackage{booktabs}
\usepackage[export]{adjustbox}
\usepackage{amsmath,amssymb,amsmath}
\usepackage{algorithm,algorithmic}
\usepackage{multirow}
\usepackage{siunitx}
\usepackage{blindtext}
\usepackage[utf8]{inputenc}
\usepackage{commath}
\usepackage{graphicx}

\usepackage{nicefrac} 
\usepackage{url}

\begin{document}
\newtheorem{prop}{\textbf{Proposition}}
\newtheorem{rem}[prop]{\textbf{Remark}}

\title{\textcolor{black}{Isotropic multi-channel  total variation framework} for joint reconstruction of multi-contrast parallel MRI }

\author{Erfan~Ebrahim~Esfahani
\thanks{The author is an independent researcher from Tehran,	Iran. Email: {erfan.ebrahim@outlook.com}. ORCID:  0000-0002-4878-3185. This study was not financially supported.}
}

\date{}
\maketitle

\begin{abstract}
\noindent \textbf{Purpose}: To develop a synergistic image reconstruction framework that exploits multi-contrast, multi-coil and compressed sensing redundancies in MRI.\\ \textbf{Approach}: Compressed sensing, multi-contrast acquisition and parallel imaging have been individually well developed but the combination of the three has not been equally well studied, much less the potential benefits of isotropy within such a setting. Inspired by total variation theory, this paper introduces a novel isotropic multi-contrast image regularizer and attains its full potential by integrating it into compressed multi-contrast multi-coil MRI. A convex optimization problem is posed to model the new  variational framework and a first-order algorithm is developed to solve the problem.\\ 
\textbf{Results}: \textcolor{black}{It turns out that the proposed isotropic regularizer outperforms many of the state-of-the-art reconstruction methods not only in terms of rotation-invariance preservation of symmetrical features, but also in suppressing noise or streaking artifacts which are normally encountered in parallel imaging methods at aggressive undersampling rates. Moreover, the new  framework significantly prevents inter-contrast leakage of contrast-specific details, which seems to be a difficult situation to handle for some variational and low-rank multi-contrast reconstruction approaches.}\\
 \textbf{Conclusion}: The new framework is a viable option for image reconstruction in fast protocols of multi-contrast parallel MRI, potentially reducing patient discomfort in otherwise long and time consuming scans.


\vspace{0.25cm}
 
\end{abstract}


\begin{textblock*}{\textwidth}(0cm,-14.5cm)

Final published version:\\
Erfan Ebrahim Esfahani, "Isotropic multichannel total variation framework for joint reconstruction of multicontrast parallel MRI," J. Med. Imag. 9(1), 013502 (2022), doi: \href{ 10.1117/1.JMI.9.1.013502.}{10.1117/1.JMI.9.1.013502.} 

\end{textblock*}

\section{Introduction, Background and Related Work}

{Magnetic} resonance imaging is perhaps \textcolor{black}{best known for} non-invasiveness, lack of hazardous radiation and excellent soft tissue visualization. However, the intrinsically slow acquisition has forced researchers to resort to various redundancies in MRI data to speed up the process. Compressed sensing (CS), multi-contrast (MC) acquisition and parallel imaging (PI) are various approaches that exploit specific types of redundancies. In CS, sparse representations of the underlying image in certain domains are used to reduce the acquired data \cite{CS MRI}. Since certain features (or lesions) are only observable in one contrast (and not all contrasts), MC imaging of the same anatomy helps to detect such a feature, where different contrasts make up for each other's deficits.  However, except for the contrast-specific features, many other details such as edges preserve the same layout across contrasts, which might be a useful redundancy. In PI \cite{PI}, multiple receive coils are placed around the object and each coil best detects the signals emitted from a specific region, depending on its sensitivity map. All regions could be imaged simultaneously, speeding up the whole sequence. 

There are two mainstream PI methods: \textit{image-domain} (or SENSE-type) approaches, pioneered by \cite{SENSE}, which are the most mathematically general strategies adaptable with arbitrary sampling patterns and yielding one single output image (hence running only once, saving computational resources), and \textit{frequency-domain} algorithms pioneered by \cite{GRAPPA}, which interpolate k-space for each coil image separately (hence running as many times as the number of coils, taking computational resources) followed by a combination of coil images. Sometimes this type of PI does not admit arbitrary sampling patterns. PI can further be divided into calibration-based methods which require a densely sampled portion of central k-space known as auto-calibration signal (ACS) region for sensitivity estimation, and calibration-free methods that do not require (but may still benefit from) such a signal \cite{STDLR,SAKE,ENLIVE,P-LORAKS}. \textit{These CS methods all consider PI alone, without MC redundancy.}

  In MC imaging, similar ideas as those in single-contrast single-coil (SC-SC) methods (e.g. \cite{LORAKS,GBRWT,Erfan}) have been developed. For instance, in \cite{haung} joint total variation (TV) and wavelet tree sparsity are explored, in \cite{BCS} group sparsity on gradients of multiple contrasts within a multi-task Bayesian framework is utilized, in \cite{JGBRWT} graph-based redundant wavelet is extended to multiple contrasts, in \cite{Ehr} new versions of TV are introduced for MC-MRI where local and directional manipulations are incorporated into TV formulation and in \cite{CDLMRI} coupled dictionary learning is used in a setting where one contrast guides the reconstruction of another. 
  \textit{All these methods consider MC imaging alone, without PI.}

In state-of-the-art literature, the methods of \cite{Itthi,J-LORAKS,SIMIT} are perhaps the closest matches for the proposed approach, in that \textit{they consider both MC and PI}. 
In \cite{SIMIT} color TV (previously developed in \cite{CTV}) and joint sparsity are used to share cross-contrast details while conventional TV and single-{contrast} sparsity are used to preserve contrast-specific features.
In \cite{Itthi}, MC-PI-CS MRI was considered via extensions of TV and total generalized variation (TGV, originally introduced in \cite{TGV MRI} for MRI). 
These extensions showed improvements over single-contrast TV and TGV, however, the experiments of the proposed paper (presented in Section \ref{sec exp}) suggest that in MC-PI setting, TGV does not offer a statistically meaningful improvement over TV. In addition, inherent drawbacks of TV and TGV such as over-smoothing, loss of contrast, staircasing \textcolor{black}{(in TV)}, leakage of a {contrast}'s exclusive details into other contrasts and lack of joint MC rotation-invariance (or joint isotropy) are still present in these extensions. 
The LORAKS family \cite{LORAKS,S-LORAKS,J-LORAKS,P-LORAKS} uses limited spatial support and smoothly varying phase assumptions (originally developed in \cite{LORAKS}, later extended to PI in \cite{P-LORAKS,S-LORAKS}) to reconstruct MR images. The MC-PI member of the family \cite{J-LORAKS} (called J-LORAKS, with ``J'' standing for joint) is similar to \cite{P-LORAKS} but concatenates multi-contrast data along coil dimension and consequently treats additional contrasts as extra PI coils. Amplification of noise and severe aliasing at higher undersampling factors are notable drawbacks inherent in the family.

In a previous work \cite{Erfan}, the \textit{rotation invariant total variation} (RITV) \textcolor{blue}{\cite{Condat}} was used to reconstruct SC-SC MR images $u\in \mathbb{C}^{n\times n}$ with 
\begin{equation} \label{RITV}
\begin{split}\text{RITV}&(u)=\min_{\{v_s:s\in S\}} \sum_{s\in S} ||v_s||_{1,2}\\ &\text{s.t. } \sum_{s\in S} L^*_sv_s - Du = 0,
\end{split}
\end{equation}
where the set \textcolor{black}{ $S =\{\leftrightarrow, \updownarrow,\bullet,+\}$} symbolized gradient field (GF) mesh grids, $D$ was the conventional 2D finite differences' operator given by $(D_1u,D_2u)(i,j)=\big(u(i+1,j)-u(i,j),u(i,j+1)-u(i,j) \big)$ and \textcolor{black}{for each $s\in S$ the linear operators $L_s$ (with adjoint $L_s^*$) interpolated the GF images $v_s=(v_{s,1},v_{s,2})\in \mathbb{C}^{2(n\times n)}$ on the new grids \cite{Condat} such that isotropy was satisfied \cite{Erfan}. Specifically, we can write out $L^*_sv_s$ for each $s\in S$ \cite{Condat,Erfan}:
\begin{equation}
\begin{split}
L_\updownarrow^* v_\updownarrow(i,j) &= \frac{1}{4} \big(  4v_{\updownarrow,1}(i,j), v_{\updownarrow,2}(i,j) + v_{\updownarrow,2}(i,j+1) +v_{\updownarrow,2}(i-1,j)+v_{_\updownarrow,2}(i-1,j+1)\big),\\
L_\leftrightarrow^* v_{\leftrightarrow}(i,j) &= \frac{1}{4} \big( v_{\leftrightarrow,1}(i,j)+ v_{\leftrightarrow,1}(i+1,j)+ v_{\leftrightarrow,1}(i,j-1)+ v_{\leftrightarrow,1}(i+1,j-1),4 v_2(i,j) \big),\\
L_\bullet^* v_{\bullet}(i,j) &= \frac{1}{2} \big(  v_{\bullet,1}(i,j)+ v_{\bullet,1}(i+1,j),  v_{\bullet,2}(i,j)+ v_{\bullet,2}(i,j+1) \big),\\
L_+^* v_+(i,j) &= \frac{1}{2} \big(  v_{+,1}(i,j)+ v_{+,1}(i,j-1),  v_{+,2}(i,j)+ v_{+,2}(i-1,j) \big).\\
\end{split}
\end{equation}}

\textcolor{black}{The proposed framework is primarily a variational approach which belongs to the same family as TV and TGV. While it inherits some intrinsic features from multi-contrast TV and TGV (e.g. admitting an arbitrariy number of contrasts $N$ which is not restricted to $N=2$ as in some other methods \cite{Ehr,CDLMRI}), it differs from them thanks to a unique, mathematically provable multi-contrast isotropic feature. As we shall see, this isotropy leads to a remarkable advantage not only over TV and TGV but also many other approaches (to the best of the author's knowledge, the potential benefits of isotropy has been ignored within MC frameworks to this day). For instance, noise and aliasing are strongly suppressed in the new framework. In addition, leakage of a {contrast}'s exclusive details into other contrasts where those details are supposed to be invisible, is suppressed much better than other MC reconstruction approaches. Moreover, unlike MC TV and TGV which seem to perform on par at least from a numerical and statistical point of view (to be discussed in Section \ref{sec exp}), the new method shows meaningful statistical improvement over other approaches.}  

 In addition to being a variational approach, the proposed method is not unlike the LORAKS family and many other recent PI methods (e.g. \cite{STDLR,ENLIVE,ESPIRiT,SAKE}), because in essence, they all seek to enforce interpolated low-rank matrix models, only in different spaces: multi-coil local k-space in the PI and LORAKS methods (inspired by Hankel structure or limited support and smooth phase) as apposed to joint pixel-wise spatial GF domain in the proposed method (inspired by rotation invariance and aligned edges). The proposed method may also be viewed as a hybrid SENSE-autocalibrating approach in PI because it mixes merits of SENSE, including admittance of arbitrary sampling patterns and one-time reconstruction, with the core advantage of autocalibration which is online estimation of sensitivity profiles directly from the undersampled (possibly noisy) data without taking extra calibration scans (the latter part is mainly played by ESPIRiT \cite{ESPIRiT}).

  In-vivo experiments will show the robustness of the proposed strategy against imperfections such as {inaccurate coil sensitivity estimations}, high undersampling factors, noise and mismatch between contrasts (due to, for example, patient motion in free-breathing sequences) where other state-of-the-art approaches simply fail. 
Convex analysis will be employed to model the proposed method, for which the Malitsky-Pock algorithm \cite{MP} is developed to obtain a solution. Even in extreme situations, the proposed algorithm needs at most one parameter to be tuned, which is yet another manifestation of robustness.
 
The rest of the paper is organized as follows: Section \ref{sec proposed framework} describes the proposed method. Section \ref{sec exp} reports the experiments and Section \ref{sec concl} makes the concluding remarks. {Appendix is also provided to cover supporting information and is referenced when needed.}

\section{Proposed Framework} \label{sec proposed framework}
 
\subsection{Theory}
Assuming square size $n\times n$ for simplicity of notation, a trivial extension of (\ref{RITV}) to a stack of multi-{contrast} square images  $(u_c)_{c=1}^N \in (\mathbb{C}^{n\times n})^N$ would be {contrast}-by-{contrast} application of the term; that is, 
$\text{RITV}((u_c)_1^N)=\sum_{c=1}^N \text{RITV}(u_c)$. However, such a naive approach is not synergistic and would not benefit from abundant redundancies and correlations among multiple contrasts. 
An isotropic regularization that exploits such information as well is more likely to succeed.

Assuming that MC images are reasonably well registered, it makes sense to \textcolor{black}{encourage} the edges to be aligned in the reconstruction. 
Indeed, it is known that reducing the number of linearly independent gradient directions at each pixel location $(i,j)$ (or equivalently, increasing their linear dependence) would lead to aligned edges in pixel domain (see Fig. 2 in \cite{TNV}). If pixel $(i,j)$ is located in the interior of a constant (or highly smooth) region of all {contrast}s, then there is no common edge at $(i,j)$ and if the gradient vectors of all contrasts are stacked in a matrix (which we call the joint pixel-wise N-{contrast} GF matrix), such a matrix would entirely be $0$, \textcolor{black}{trivially leading to rank $0$}.  Now, if $(i,j)$ is on a common edge in all N {contrast}s, then all GF vectors are (nearly) the same (up to an unimportant sign flip), leading to a rank of (nearly) $1$. Finally, assume the most complicated case where a lesion is observed in one contrast but invisible in the others and $(i,j)$ is located on an edge of such a lesion. Then there is a non-trivial gradient vector at $(i,j)$ in the lesion's contrast of origin but such a gradient does not exist in the other contrasts and the gradient vectors at $(i,j)$ in those contrasts are (nearly) $0$ and highly linearly-dependent. \textcolor{black}{This case also corresponds to rank $1$}. Therefore, a low-rank behavior is anticipated in the joint pixel-wise GF matrix at all times. Fig. \ref{MC_phantom} gives an example of this idea. \textcolor{black}{We also remark that in this context, the maximum rank that the matrix could have is $2$ (becasue there are $2$ dimensions). Therefore, although the rank $2$ might seem low, for our purposes this rank is actually high and even unacceptable because it corresponds to the case where images are misaligned and poorly registered (again, see Fig. 2 in \cite{TNV}). Thus, theoretically, our regularizer has to be heavy enough to avoid rank $2$.  }   

 \begin{figure}\centering
\includegraphics[width=0.6\linewidth]{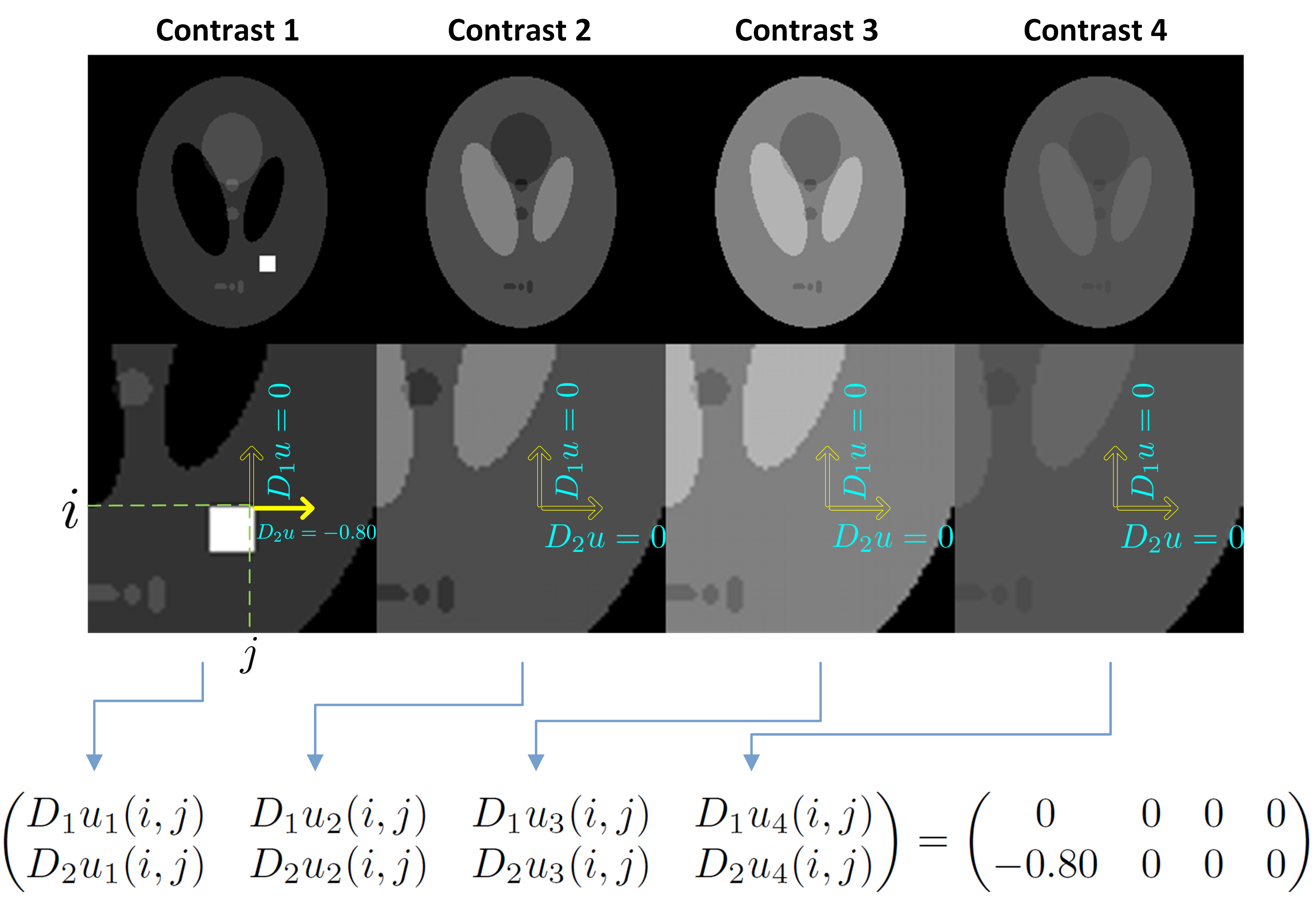}
	\caption{A simple simulation experiment involving a synthetically generated 4-contrast Shepp-Logan phantom, showing low-rankness of the joint gradient matrix at a pixel $(i,j)$ located on the top right corner of a feature exclusive to the first contrast. The hollow arrows on the figure indicate a null gradient component while a solid arrow represents a non-zero component. In this experiment $n=200$ and $(i,j) = (140,130)$. The computation also seems to suggest sparseness of the matrix. }
	\label{MC_phantom}
\end{figure}

 There are plenty of ways of enforcing low-rankness (we will discuss this \textcolor{black}{shortly}), however, we choose to model this behavior by direct thresholding of singular value decomposition (SVD), minimizing the nuclear norm of the pixel-wise joint N-{contrast} GF matrix over the complex GF variables interpolated by the operators $L_s,\,s\in S$. \textcolor{black}{This approach leads to a joint MC rotation-invariance property.}  
 The classical TV function in calculus is isotropic in the continuous domain for any rotation, but achieving a similar result in the discrete domain has proven to be elusive (\textcolor{black}{see \cite{Condat} and references cited therein}.) Nevertheless, it would be interesting to have a discrete version of this property at least for modulo $\pi/2$ rotations. In \cite[Figure 1, Table 1]{Condat} isotropic failure of many TV discretizations (including the so-called ``isotropic" discrete TV well-known in medical imaging literature) is established for flipping and $\pi/2$ rotations. In \cite[Figure 1, Figure 6]{Erfan} the isotropic failure of TV and TGV for $\pi/2$ rotations in CS MRI reconstruction was also observed.
 MR images can have many symmetrical and isotropic features and guaranteed isotropy will preserve such features, significantly improving the performance, even more so in MC settings. Based on these discussions, the following extension of (\ref{RITV})
is proposed ($ c\in\mathcal{C}:=\{1,\cdots,N\}$ and $||\cdot||_{\star}$ is the nuclear norm):
\begin{equation} \label{NRITV}
\begin{split}\text{NRITV}((u_c)_1^N):=&\min_{\{v_s^c\}} \sum_{s\in S} \sum_{i,j}||(v_s^1(i,j) \cdots v_s^N(i,j))||_{\star}\\ \text{s.t. } &\sum_{s\in S} L^*_sv_s^1 - Du_1 = 0,\\
&\vdots\\&\sum_{s\in S} L^*_sv_s^N - Du_N = 0.
\end{split}
\end{equation}
The isotropy of (\ref{NRITV}), which is one of the main theoretical results of this section, is asserted at once:
\begin{prop}\label{prop1}
Let $\mathcal{R}$ be the operator that rotates each {contrast} by $\ang{90}$. Then $\text{NRITV}(\mathcal{R}(u_c)_1^N) = \text{NRITV}((u_c)_1^N)$.
\end{prop}
\begin{proof}
For any 2D digital image $u\in\mathbb{C}^{n\times n},\, \mathcal{R}u(j,i)=u(i,n-j+1)$. With this in mind,
 assume that $\big((v_s^c)_{c \in \mathcal{C}}\big)_{s\in S}\in \big((\mathbb{C})^{n\times n}\big)^{4N}$ is a solution to problem (\ref{NRITV}), which is known to exist by convexity. Consider the problem 
   \begin{equation} \label{2proposed model}
   \begin{split}
   \min_{\{\bar{v}_s^c:s\in S, c\in \mathcal{C}\}}  
   &\lambda\sum_{s\in S} ||(\bar{v}_s^c)_{c=1}^N||_{1,\star}\\
   \text{s.t. } \sum_{s\in S}  L^*_s &\bar{v}_s^c -  D(\mathcal{R}u_c) = 0;\, \forall c\in \mathcal{C}.
   \end{split}
   \end{equation}
   With some algebraic manipulation, the following solution can be seen to satisfy the constraints in (\ref{2proposed model}) for each $c \in \mathcal{C}$:

\begin{equation*}
	\begin{split}
		&\bar{v}_\bullet^c =(\bar{v}^c_{\bullet,1}, \bar{v}^c_{\bullet,2})=(-\mathcal{R}{v}^c_{\bullet,2}, \mathcal{R}{v}^c_{\bullet,1}),\\ &\bar{v}^c_\updownarrow(i,j) =(\bar{v}^c_{\updownarrow,1}(i,j), \bar{v}^c_{\updownarrow,2}(i,j))=(-\mathcal{R}{v}^c_{\leftrightarrow,2}(i+1,j), \mathcal{R}{v}^c_{\leftrightarrow,1}(i+1,j)),\\
		&\bar{v}^c_\leftrightarrow=(\bar{v}^c_{\leftrightarrow,1}, \bar{v}^c_{\leftrightarrow,2})=(-\mathcal{R}{v}^c_{\updownarrow,2}, \mathcal{R}{v}^c_{\updownarrow,1}),\\ &\bar{v}^c_+=(\bar{v}^c_{+,1}, \bar{v}^c_{+,2})=(-\mathcal{R}{v}^c_{+,2}, \mathcal{R}{v}^c_{+,1}).
	\end{split}
\end{equation*}
It remains to show that 
\begin{equation}\label{eq13}
 \sum_{s\in S} ||({v}_s^c)_{c=1}^N||_{1,\star}=
\sum_{s\in S} ||(\bar{v}_s^c)_{c=1}^N||_{1,\star}.
\end{equation}
For $s=\bullet$, we have
\begin{equation}
\begin{split}
&||(\bar{v}^c_{\bullet})_{c=1}^N||_{1,\star} = \sum_{i,j=1}^n 
\bigg\Vert\begin{pmatrix}
-\mathcal{R}v^1_{\bullet,2}(j,i) & \cdots& -\mathcal{R}v^N_{\bullet,2}(j,i)  \\ \mathcal{R}v^1_{\bullet,1}(j,i) & \cdots
&\mathcal{R}v^N_{\bullet,1}(j,i)\end{pmatrix} \bigg\Vert_{\star}\\&=
\sum_{i,j} 
\bigg\Vert\begin{pmatrix}
-v^1_{\bullet,2}(i,n-j+1) & \cdots& -v^N_{\bullet,2}(i,n-j+1)  \\ v^1_{\bullet,1}(i,n-j+1) & \cdots
&v^N_{\bullet,1}(i,n-j+1)\end{pmatrix} \bigg\Vert_{\star}
\\&=
\sum_{i,j} 
\bigg\Vert\begin{pmatrix}
-v^1_{\bullet,2}(i,j) & \cdots& -v^N_{\bullet,2}(i,j)  \\ v^1_{\bullet,1}(i,j) & \cdots
&v^N_{\bullet,1}(i,j)\end{pmatrix} \bigg\Vert_{\star}
=||({v}^c_{\bullet})_{c=1}^N||_{1,\star},
\end{split}
\end{equation}
where the penultimate equality was deduced from indifference of finite sums with respect to rearrangement of summands, and the final equality was thanks to the definition of nuclear norm and indifference of the characteristic polynomial of $2\times N$ matrices with respect to row-wise sign flips and position swaps. Similarly, one can see that
$
||(\bar{v}^c_{+})_{c=1}^N||_{1,\star}=||({v}^c_{+})_{c=1}^N||_{1,\star}$, 
$||(\bar{v}^c_{\updownarrow})_{c=1}^N||_{1,\star} = ||({v}^c_{\leftrightarrow})_{c=1}^N||_{1,\star}$ and
$||(\bar{v}^c_{\leftrightarrow})_{c=1}^N||_{1,\star}=
||({v}^c_{\updownarrow})_{c=1}^N||_{1,\star}$. Eq. (\ref{eq13}) follows by summing these terms.
\end{proof}

\begin{rem}\label{rem1}
 Invoking the result above twice in a row gives isotropy for $\ang{180}$, and 3 times in a row for $\ang{270}.$      
\end{rem}

Simply put, (\ref{NRITV}) interpolates and reshapes the GF in each individual {contrast} through the constraints, and then sparsifies the set of linearly independent pixel-wise gradient vectors through the cost function, which would eventually force isotropy, edge alignment and preservation of contrast-specific details.
\textcolor{black}{The idea of low-rank modeling is quite popular in current MRI literature (in the Experiments section we will provide comparison with as many such works as we can). For example, crude nuclear-norm gradient-coupling has been suggested for TV and TGV (see \cite{TNV,MR-PET} and references therein) for different applications. NRITV improves upon these regularizers by means of a redefined GF discretization and gauranteed isotropy. This gradient field ($v_s,\,s\in\{\updownarrow\,\leftrightarrow,\bullet,+\}$) is four times as dense as those in TV and TGV, and it is this refinement in GF that allows for spatial isotropy, which was not possible in TV and TGV \cite{Erfan}}. \textcolor{black}{Other low-rank penalties have been published that do not require the direct computation of SVDs; for example, matrix factorizations (PQ and UV in \cite{STDLR} and \cite{ENLIVE}) have been suggested. However, in this work, we present a direct SVD thresholding approach. This results in a convex optimization problem, which admits a global solution and is required for our proof of Proposition \ref{prop1}. }
Proper validation of NRITV's merits will be presented in Section \ref{sec exp} where we will observe that NRITV fulfills all of its theoretical promises in practice. In addition, it turns out that the proposed approach outruns \textcolor{black}{(see Fig. \ref{times})} and outperforms \textcolor{black}{(see Section \ref{sec exp})} the SVD-free method of \cite{STDLR}. We are now ready to put forth the new MC-PI-CS MR image reconstruction model:
\begin{equation} \label{proposed model}
\begin{split}
\min_{\{u_c \geq 0,v_s^c\}}
\frac{1}{2}&\sum_{c=1}^N\sum_{p=1}^P||\mathcal{F}_pu_c-b_{c,p}||_2^2 + 
\lambda\sum_{s\in S} ||(v_s^c)_{c=1}^N||_{1,\star}\\
    \text{s.t. } &\sum_{s\in S} L^*_sv_s^c - Du_c = 0; \quad c\in\{1,\cdots,N\},
\end{split}
\end{equation}
where $P$ is the number of parallel imaging coils, $b_{c,p}\in \mathbb{C}^{n\times n}$ is the zero-filled subsampled k-space data for coil $p$ and contrast $c$, $\mathcal{F}_p$ is the system matrix involving a partial Fourier operator and sensitivity modulation for coil $p$ and $\lambda>0$ 
is a balancing parameter. The restriction $u_c\geq 0$ is a shorthand for $\text{Re }(u_c)\geq0$ and $\text{Im }(u_c)\geq0$ (there is not much meaning to negative pixel values, be it in the real or imaginary part).

  \begin{figure*}\centering   \includegraphics[width=0.88\textwidth]{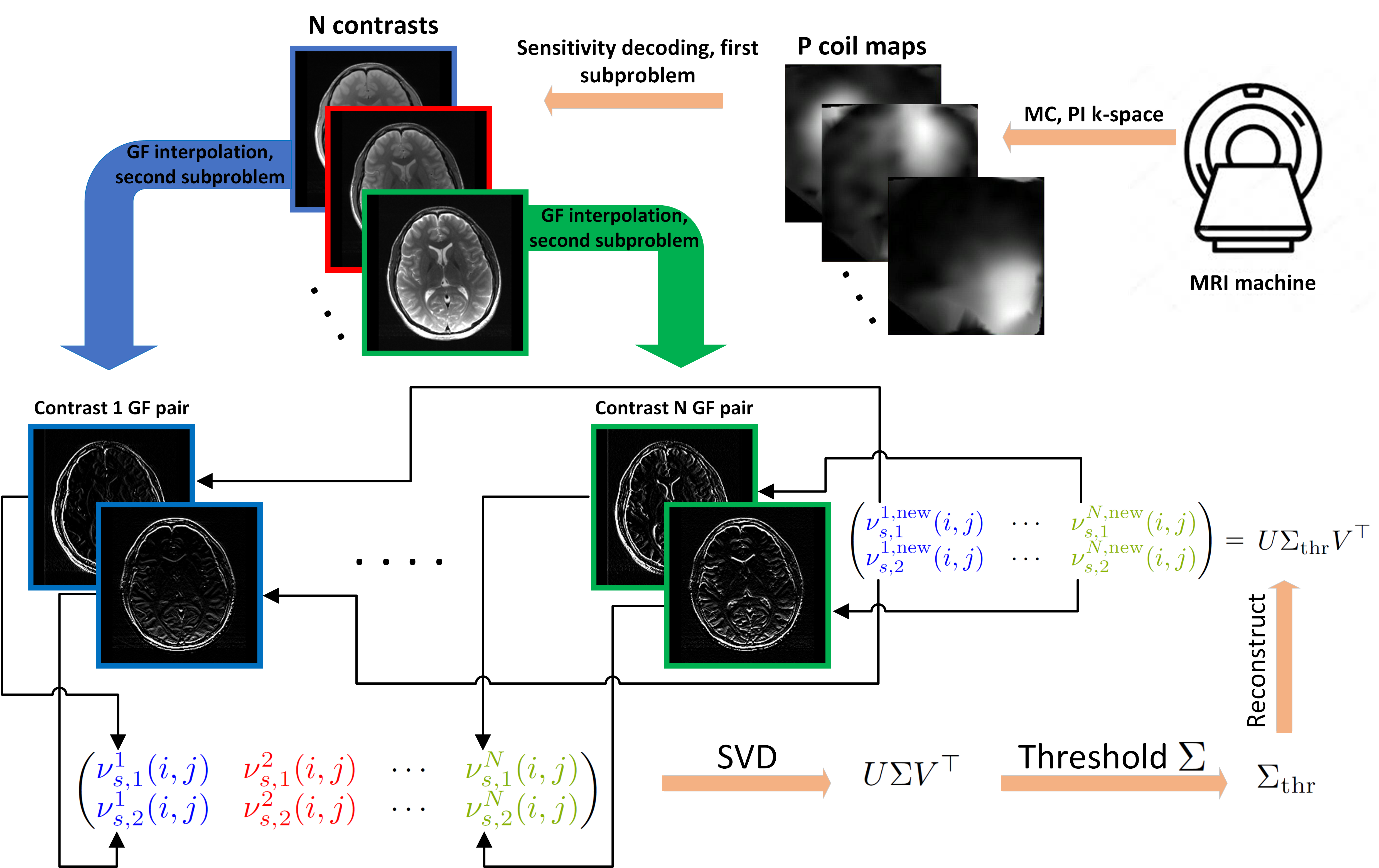}   \caption{Schematic representation of the $u$ and $v$ subproblems of the proposed algorithm. The sensitivity-encoded multi-contrast data are decoded back to the image domain using coil maps (upper part of the figure). Thereupon, the gradient fields of all $N$ contrasts are interpolated using $L_s$, $s\in S$, after which a joint gradient matrix is extracted from all contrasts at every pixel $(i,j)$. The SVD of this matrix is then taken, followed by soft-thresholding of the singular values. The matrix is subsequently reconstructed and each new GF value returns to its original position. The extraction-thresholding loop runs until $i,j$ and $s$ are exhausted. 
	}\label{fig1}       \end{figure*}

Problem (\ref{proposed model}) is a convex minimization, which can be recast as a convex-concave saddle point problem:

\begin{equation}
	\begin{split}
	\min_{u_c,v_s^c} \max_{h_c,r_{c,p}}\, \lambda \sum_{s\in S} &||(v_s^c)_{c=1}^N||_{1,\star}+ \sum_c \sum_p
		\langle \mathcal{F}_p u_c -b_{c,p}, r_{c,p}\rangle\\+\sum_c &\delta_{\mathbb{R}^n_+}(u_c) - \frac{1}{2} \sum_c \sum_p||r_{c,p}||_2^2 
		+ \sum_c\langle \sum_{s\in S} L^*_sv_s^c -Du_c, h_c \rangle, \label{proposed model 2}
	\end{split}
\end{equation}
where $r_{c,p}\in \mathbb{C}^{n\times n}$ and $h_c \in \mathbb{C}^{2(n\times n)}$ are frequency and GF  elements and $\delta_A(u)$ is the indicator function of the convex set $A$. The primal and dual variables $x$ and $y$ and the corresponding cost functions $g(x)$ and $f^*(y)$ are respectively set as $x^{\top}:=(u_c,v^c_{\updownarrow},v^c_{\leftrightarrow},v^c_{\bullet}, v_+^c)_{c=1}^N$, $y^\top:=\big((r_{c,p},h_c)_{c=1}^N\big)_{p=1}^P$,
\begin{align}
	g(x) &:=\sum_c \delta_{\mathbb{R}^n_+}(u_c)+ \lambda \sum_{s\in S} ||(v_s^c)_{c=1}^N||_{1,\star},\label{g}\\
	f^*(y) &:= \frac{1}{2}\sum_c \sum_p||r_{c,p}||_2^2 +\sum_c \sum_p\langle b_{c,p},r_{c,p} \rangle ,\label{f*}
\end{align}    
together with 
\begin{equation}
	K  := 
	\begin{pmatrix}
		\mathcal{F} & 0\\
		- \mathcal{D} & \mathcal{L^*} 
	\end{pmatrix},
\end{equation} 
where $\mathcal{D}$, $\mathcal{L^*}$ and $\mathcal{F}$ are all $N$-block-diagonal matrices whose diagonal entries are respectively  $D$,
$[L^*_{\updownarrow} \, L^*_{\leftrightarrow} \, L^*_{\bullet}\,  L^*_+]$ and the concatenation of all coils' system matrices.
We can now develop the Malitsky-Pock algorithm \cite{MP} for (\ref{proposed model 2}).
  \subsection{Algorithm}

 \subsubsection{The $u_c$ subproblem} 
Once again assuming that $u_c$ stands for both real and imaginary parts, update for  $u_c$ is given by:
 \begin{equation}
 	\begin{split}
 u^{k+1}_c &= \text{prox}_{\tau_{k}\delta_{\mathbb{R}^n_+}} \Big(u^{k}_c-\tau_{k}\big(\sum_p\mathcal{F}_p^*(r^{k}_{c,p})-D^*(h^{k}_c)\big) \Big) = \max \big\{ u^{k}_c-\tau_{k}\big(\sum_p\mathcal{F}_p^*(r^{k}_{c,p})-D^*(h^{k}_c)\big),0\big\}.\label{u solution}
	\end{split} 
\end{equation}
Data acquired from k-space are sensitivity-encoded in the Fourier domain. This subproblem thus carries out sensitivity decoding in image domain. However, to get a good result, senitivity maps used in $\mathcal{F}_p^*$ must be sufficiently accurate. In most of the experiments, we will approximate the maps by an estimation algorithm such as ESPIRiT \cite{ESPIRiT}.
 \subsubsection{The $v^c_s$ subproblem}  
 This part is the core of the proposed algorithm and the crossroads where all $N$ {contrast}s exchange information and joint GF update takes place. It is well known that for a separable function, the proximal operator decouples into a concatenation of the proximities of the function components \cite[Remark 6.7]{Beck}. Invoking this result multiple times and setting $(\nu_s^{c,k})_{c=1}^N := (v_s^{c,k})_{c=1}^N - \tau_{k} (L_s(h^{k}_c))_{c=1}^N$, we obtain for each $s\in S$:
\begin{align}\label{v update}
	(v_s^{c,k+1})_{c=1}^N &= \text{prox}_{\tau_{k} \lambda ||\cdot||_{1,\star} } ((\nu_s^{c,k})_{c=1}^N) \notag
	= \Big(\text{prox}_{\tau_{k} \lambda||\cdot||_{\star}} \big(\nu_s^{c,k}(i,j)\big)_{c=1}^N \Big)_{i,j=1}^n\notag\\ 
	& = \Bigg( \text{prox}_{\tau_{k} \lambda||\cdot||_{\star}}
	\begin{pmatrix}
		\nu^{1,k}_{s,1}(i,j) & \cdots & \nu^{N,k}_{s,1}(i,j)\\
		\nu^{1,k}_{s,2}(i,j)& \cdots & \nu^{N,k}_{s,2}(i,j)
	\end{pmatrix} \Bigg)_{i,j=1}^n.\end{align}
Therefore, for fixed $i,j$ and $s$, SVD of the joint pixel-wise GF matrix should be computed, followed by soft-thresholding of the singular values by $\tau_k \lambda$ and reconstructing the new GF matrix. The procedure is schematically represented in Fig. \ref{fig1}.

\begin{algorithm} 
	\caption{Proposed method for compressed multi-contrast parallel MRI}
	\begin{algorithmic}[1]\label{proposed algorithm}
		\renewcommand{\algorithmicrequire}{\textbf{Initialization:}}
		\renewcommand{\algorithmicensure}{\textbf{Output:}}
		\REQUIRE Choose $\theta_0 = 1,\, u^0_c = u^\text{zf}_c,\, v_s^{c,0}=0$  $ h^0_c=0,\,r^0_c=0\text{ for } s\in S\text{ and } c\in \mathcal{C}, $  $ \tau_0>0, \beta>0,\, \mu \in (0,1)$ and $ \delta\in (0,1)$. Define $\mathcal{P} := \{1,\cdots P\}.$ 
		\\ \textit{\textbf{While convergence criterion not met, repeat:}}
		\STATE \label{step 1} $\text{Obtain } u_c^{k+1} \text{ from (\ref{u solution}) for each }c\in \mathcal{C} ;$
		\STATE $\text{Obtain }( v_s^{c,k+1})_{c=1}^N \text{ jointly from (\ref{v update}) for each } s\in S ;$
		\STATE Choose $\tau_{k+1} \in [\tau_{k}, \tau_{k}\sqrt{1+\theta_{k}}];$
		\\ \textit{\textbf{Linesearch:}}
		\STATE $\theta_{k+1} = \frac{\tau_{k+1}}{\tau_{k}};$ \label{linesearch}
		\STATE $\bar{u}^{k+1}_c = u_c^{k+1} + \theta_{k+1}(u^{k+1}_c-u^{k}_c),\, \forall c \in \mathcal{C};$
		\STATE $\bar{v}^{c,k+1}_s = v_s^{c,k+1}+ \theta_{k+1}(v_s^{c,k+1}-v_s^{c,k}),\,\forall s\in S,\, c\in \mathcal{C};$
		\STATE $ \text{Obtain }r^{k+1}_{c,p} \text{ from (\ref{r update}) for each } c\in \mathcal{C} \text{ and } p\in \mathcal{P};$
		\STATE $ \text{Obtain } h^{k+1}_c
		\text{ from (\ref{h update}) for each } c\in \mathcal{C};$
		\IF {$\sqrt{\beta}\tau_{k+1}||  \big((\mathcal{F}^*_{p}(r^{k+1}_{c,p}-r^{k}_{c,p}))_{c,p}, (L_s (h_c^{k+1}-h_c^{k}))_{s,c}\big)||_2$ $\leq \delta||((r^{k+1}_{c,p}-r^{k}_{c,p},h^{k+1}_c-h^{k}_c))_{c,p}||_2  $}
		\STATE Return to step \ref{step 1} (break linesearch),
		\ELSE
		\STATE Set $\tau_{k+1} = \mu \tau_{k+1}$ and return to step \ref{linesearch} (apply another iteration of linesearch).
		\ENDIF
		\ENSURE  Reconstructed stack of multi-contrast MR images $(u_c)_{c=1}^N$, solution to (\ref{proposed model}).
	\end{algorithmic} 
\end{algorithm}

\subsubsection{The $y$ subproblem}
This subproblem encompasses $r_{c,p}$ and $h_c$ updates, given respectively by
\begin{align}
	r^{k+1}_{c,p} &=\text{prox}_{\beta \tau_{k+1} (\frac{1}{2}||\cdot||_2^2 +\langle \cdot,b_{c,p} \rangle ) }(r^{k}_{c,p}+\beta \tau_{k+1} \mathcal{F}_p(\bar{u}_c^{k+1})) =\frac{r^{k}_{c,p}+\beta \tau_{k+1} \mathcal{F}_p(\bar{u}_c^{k+1}) - \beta \tau_{k+1}b}{1+\beta \tau_{k+1}}\label{r update} , \\
	h^{k+1}_c &= h^{k}_c + \beta \tau_{k+1} (-D(\bar{u}_c^{k+1}) + \sum_{s\in S}L_s^*(\bar{v}_s^{c,k+1}))\label{h update}.
\end{align}

Fig. \ref{fig1} gives a simplified representation of some major steps in the proposed algorithm.
A more rigorous and formal  description is detailed in Algorithm \ref{proposed algorithm}.  

\section{Experiments}\label{sec exp}

\subsection{Setup} \label{sec setup}
\subsubsection{Hardware and software}\label{Hardware}In this section the proposed framework will be evaluated. Real-valued phantoms as well as complex-valued in-vivo data are considered for experiments. 
 {Based on software availability and reproducibility, the following state-of-the-art methods have been considered for comparisons: from SC-SC literature \textbf{RITV+BM3D} (briefly RITV) \cite{Erfan}; from PI literature \textbf{SAKE} \cite{SAKE}, \textbf{L1-ESPIRiT} \cite{ESPIRiT}, \textbf{SENSE-LORAKS} (briefly S-LORAKS) \cite{S-LORAKS}, \textbf{STDLR-SPIRiT} \cite{STDLR}; from MC-PI literature \textbf{J-LORAKS} \cite{J-LORAKS}, \textbf{MC-TV-SENSE-nuclear} (briefly TV-nuc) \cite{Itthi,TNV} and \textbf{MC-TGV-SENSE-nuclear} (briefly TGV-nuc)} \cite{Itthi,MR-PET}. TV-nuc and TGV-nuc were based on the codes of \cite{Itthi}, but the anistropic vectorial TV and TGV in \cite{Itthi} were respectively replaced by nuclear-norm TV \cite{TNV} and nuclear-norm TGV \cite{MR-PET} to boost TV and TGV performances.
 The data and software reproducing the proposed method are publicly available at \cite{Mend}.   
The codes for other methods were also shared by their authors.   
Experiments were conducted in MATLAB R2020b (MathWorks, Natick, Massachusetts, USA) mostly on a laptop with an AMD FX-7600 Radeon R7 CPU at 2.70 GHz clock speed, AMD R9 M280X GPU with 4GB of memory and 8GB of RAM. However, due to higher computational demand by STDLR-SPIRiT at size $n=200$, two more machines with 16 and 32 GB of RAM were devoted exclusively to this method at $n=200$. The quality of reconstructions is quantified by SSIM (MATLAB's built-in function) and $\text{SNR}(u,u_{\text{true}}) := 20\log\frac{||u||_2}{||u-u_{\text{true}}||_2}$.

\subsubsection{Sampling and calibration}
Cartesian sampling patterns with fairly aggressive reduction factors of $R=5,7$ and $9$ have been primarily used in the experiments (with $R:= \nicefrac{n^2}{\#\text{samples}}$, so that for example, $R=7$ corresponds to about $14\%$ sampling). This is for two reasons: first, Cartesian patterns are by far the most popular patterns in actual clinical applications and we want to remain as close as possible to clinical protocols, and secondly, with highly incoherent patterns \cite{CS MRI} like variable density random, spiral or Poisson disk, most MC-PI methods produce almost ground-truth quality results (even at aggressive reduction factors), \textcolor{black}{in which case comparison of the methods would not show much difference between them. For example, multi-contrast TV and TGV in \cite{Itthi} are both reported to produce SSIM above $0.98$ for Poisson disk undersampling even at $R=15$.} Fig. \ref{Supp_radial} compares reconstructions from a non-Cartesian pattern. {The Cartesian patterns were made by fixing $40\%$ of the samples symmetrically at central k-space (the ACS region) and selecting the remaining $60\%$ randomly along the phase encode direction.} {A typical pattern is shown in the top right corner of Fig. \ref{lp}. Although shifting the sampling pattern across contrasts is an option, such a strategy did not show worthwhile empirical improvements for any of the methods considered here, hence we always use a fixed pattern.}
{In all experiments, perfect coil sensitivity maps were used in SENSE-LORAKS, MC-TV-SENSE-nuc and MC-TGV-SENSE-nuc to give these methods an extra advantage. However, the proposed method mostly used only a harsh estimation of the sensitivity maps provided by ESPIRiT \cite{ESPIRiT}. This is because we wanted to test the robustness of our approach against inaccuracies of sensitivity estimations which exist in actual situations as often as not. To asses  inaccuracy of the estimated maps, relative $l_2$ norm error (RLNE) is used, defined as usual: $\text{RLNE}(\tilde{\mathcal{S}}) = \frac{||\mathcal{S}-\tilde{\mathcal{S}}||_2}{||\mathcal{S}||_2}$ where $\mathcal{S}$ and $\tilde{\mathcal{S}}$ are vector-stacked perfect and estimated sensitivity maps, respectively. The maps used in NRITV were mostly estimated real-time by ESPIRiT directly from the ACS region in the acquired undersampled (noisy) data without simulating separate calibration scans. This harsh approach resulted in $14\%$ to $40\%$ RLNE in sensitivity estimation for NRITV (as opposed to $0\%$ for the SENSE-type methods mentioned earlier). 
} 

\subsubsection{Parameter settings}
 \textcolor{black}{Parameters of compared methods were carefully tuned to (nearly) \textcolor{black}{optimize} their performance based on SSIM, SNR and visual perception. 
 Tuned parameters for compared methods in each experiment are detailed in the appendix tables (and any parameters not reported there keep their script defaults).  Moreover, the process of parameter optimization is exemplified in Fig. \ref{param_opt}.}
 The parameters for the proposed method are $\{\beta = 4\times 10^{-5}, \lambda = 7\times 10^{-5}, \mu = 0.7, \delta = 0.99 ,\theta_0=1,\,\text{max iterations}=  300\}$ but not rigorously optimized. Unless dealing with noisy data, in which case only $\lambda$ needs readjustment, the user is not required to change these settings. 
  \textcolor{black}{However, the interested user may optimize the parameters for this method, for example, using the NRITV software available at \cite{Mend}.}

 \subsection{Results} \label{results}

 \subsubsection{Phantom experiments}  
 
 \begin{figure*}\centering
 \includegraphics[width=\linewidth]{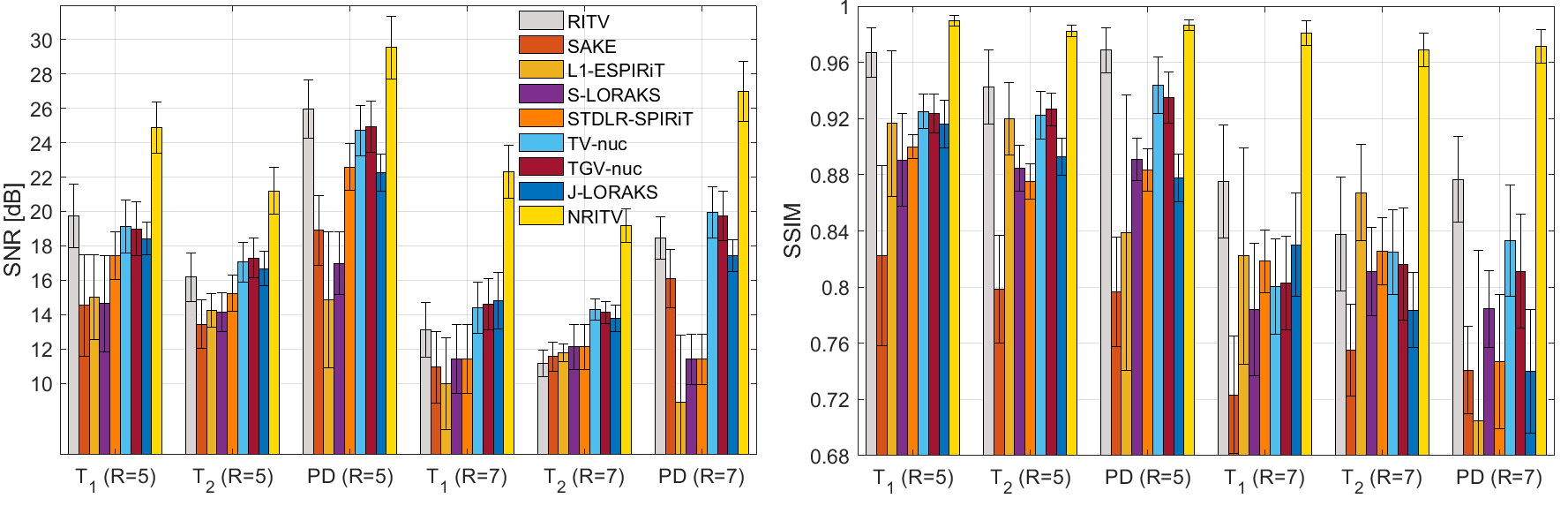}
 \caption{Performance of all compared works in the phantom data in terms of SNR (left) and SSIM (right).}
 \label{BW_bars}
 \end{figure*}
 \begin{figure*} 
	\includegraphics*[width = \textwidth]{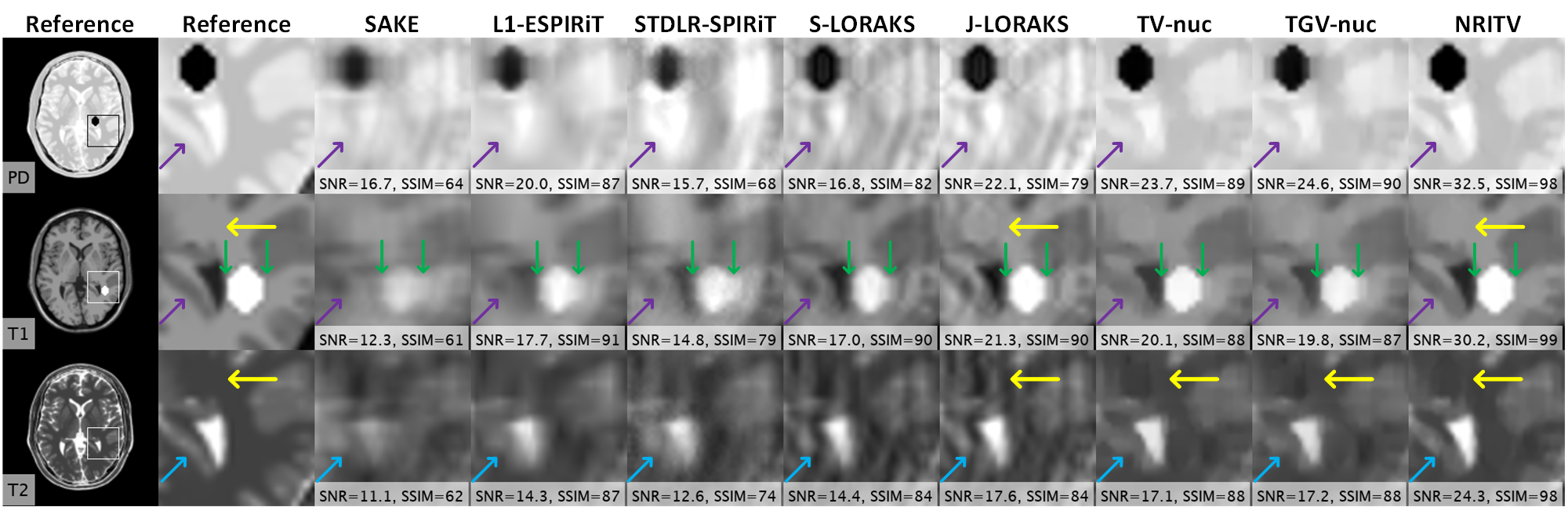}
	\caption{Magnified views of the results for the brain phantom at $R=7$. Prominent errors in the competing methods and their corrections with the proposed one are marked by arrows. {Appendix Fig. \ref{Supp_Lesion} shows full size FOVs and error maps for this experiment.}}
	\label{showpice} \end{figure*}  
\begin{figure*}\centering
	\includegraphics[width=0.7\linewidth]{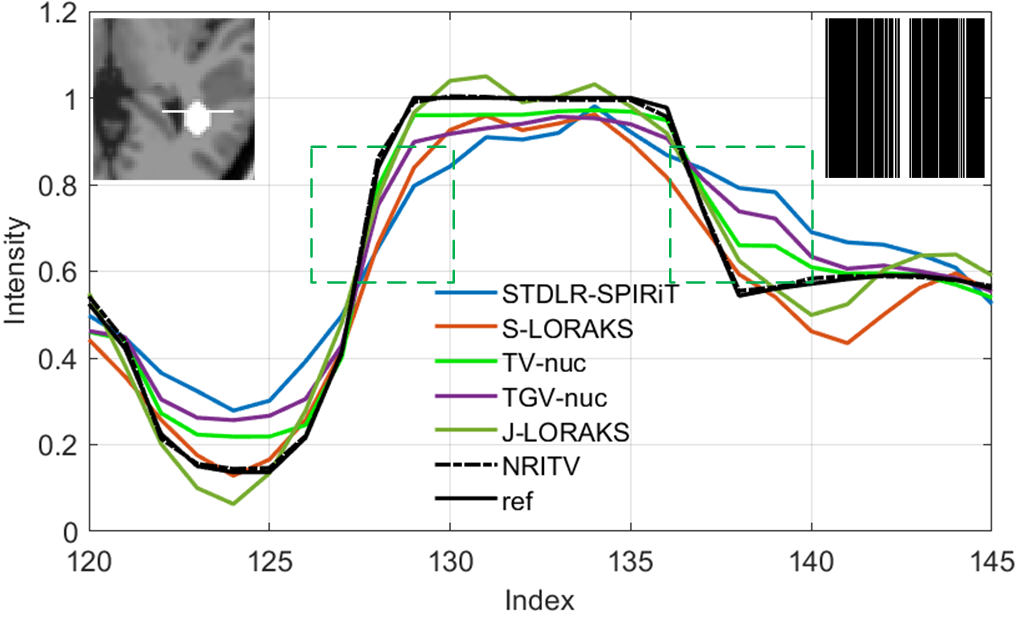}
	\caption{Line profiles through the $T_1$-only lesion in the top 6 results of Fig. \ref{showpice}. The dashed green box on the left shows a sharp transition in intensity on the left side of the lesion, perfectly captured by NRITV, fairly by the others;  and the right box shows the mirror transition on the right side of the lesion, again perfectly captured by NRITV but poorly by the others. The dashed geen boxes correspond to the green arrows in Fig. \ref{showpice}. {The sampling pattern used in this test is shown on the top right corner.}}
	\label{lp}
\end{figure*} 
 
 In order to form an initial idea on the performance of the proposed framework, 60 multi-contrast slices distributed into triplets of $T_1$-weighted, $T_2$-weighted and proton-density-weighted scans of the normal brain were obtained from a realistic brain phantom \cite{BW}. The triple-contrast 8-coil phantom data were resized to $n=80$ for computational practicality, retrospectively undersampled at $R=5,7$ and reconstructed by all 9 methods. {NRITV used inaccurate ESPIRiT-based sensitivity maps with an average RLNE of $23\%$ over the whole dataset (using only the $6$ and $4$ ACS lines respectively at $R=5$ and $R=7$ with fixed calibration settings). Optimized parameters for competing methods are detailed in Table \ref{Table_BWR5} and Table \ref{Table_BWR7}.} The entire output data were too large to be shown. Instead, average $\pm$ standard deviation values of SNR and SSIM for all considered methods are reported in Fig. \ref{BW_bars}, which can be interpreted as follows:
\begin{itemize}
\item
 Although RITV does not benefit from any form of multi-coil or multi-contrast redundancy, it is comparable with such methods and outperforms some of them. 
 \item All PI and MC-PI methods (excluding NRITV)  \textcolor{black}{show some instabilities}, each doing well in some experiments by some metric while under-performing in another setting. The comparison of these methods may be considered inconclusive as none strongly outperforms the others.
 \item The proposed method on the other hand \textcolor{black}{strongly} outperforms all other methods.
 \end{itemize}

\begin{rem}\label{rem2}
{Henceforth we shall focus on visual comparison of the methods, first on phantom, then on in-vivo data. As such, the size is set to $n=200$. To save space, only some magnified regions in reconstructed images are selected and visually compared here. Full-size FOVs and other details can be viewed in the appendix.}
\end{rem}
 Fig. \ref{showpice} illustrates an example from the BrainWeb phantom test set. A sampling pattern with the reduction factor $R=7$ {(see the top right corner in Fig. \ref{lp})} along with 4 coils were used. {The coil maps used in NRITV were directly obtained from the ACS region using ESPIRiT with $19\%$ RLNE.} In order to simulate the situation where a lesion is observable in one contrast but invisible in others, a black hexagon in the PD contrast and a white hexagon in the $T_1$ contrast were superimposed to the images. Results show several blatant errors in compared methods, some of which are marked by arrows. SNR (in dB) and SSIM (in $\%$) are also attached to the bottom of each reconstruction.

 The purple and blue arrows in Fig. \ref{showpice} point to the posterior horn of lateral ventricle and the nearby grey matter which are over-smoothed and poorly captured by  compared methods but flawlessly with NRITV.

 The green arrows in the $T_1$ results in Fig. \ref{showpice} point to the side edges of the $T_1$-only lesion. While some of the compared methods show acceptable results on the left edge, they all blur the right edge. In other words, the right-side edge of the lesion is not as sharply reconstructed as the left one. However, with NRITV the edges are equally sharp on all sides. This shows the significance of isotropy in regularization, because the lesion is symmetric and invariant under $\ang{180}$ rotation (see also Remark \ref{rem1}), hence a similar result on both sides is expected. This observation is also in line with those made by Condat \cite{Condat}, where it was discussed that TV does not sufficiently penalize oblique edges, leading to non-isotropy. The same could be said of TGV, because the GF used in TGV is essentially the same as TV. Other methods lead to a similar artifact. Fig. \ref{lp} highlights this further with line profiles. Due to symmetry in the lesion, intensity transitions when the line enters the lesion from the left side in the reference image (about index $127$, marked by the left green dashed box) is almost the same as when the line leaves the lesion on the right side (about index $137$, marked by the right green dashed box). The compared methods fail to capture the intensity jump in the right side due to a lack of isotropy, while NRITV is almost indistinguishable from the symmetric reference profile.

  The yellow arrows in Fig. \ref{showpice} 
  point to certain locations in the images where a leakage from another contrast is expected. It is evident that the black lesion from the PD contrast has leaked slightly into J-LORAKS' $T_1$ result. On the other hand, TV-nuc and TGV-nuc have suppressed this leakage in $T_1$ contrast almost as effectively as NRITV. However, the same lesion has left some faint residuals in the $T_2$ result with TV-nuc and TGV-nuc which are cleared out by NRITV.

 Fig. \ref{showpice} exemplifies some of NRITV's merits, namely, retention of contrast-specific features, isotropy and substantial improvement over other methods despite sensitivity maps being perfectly known to them but inaccurately known to NRITV. { Full-size results for all methods along with error maps can be viewed  in Fig. \ref{Supp_Lesion}.  Optimized parameters for compared methods are reported in Table \ref{Table_phantom}. }

\subsubsection{In-vivo experiments}
To further demonstrate usefulness of the proposed strategy, experiments on in-vivo MR data are also carried out. All images in this subsection are complex-valued (magnitudes are shown). Various anatomies have been chosen to investigate versatility of the proposed approach.

 \begin{figure*} 
     \includegraphics[width = \textwidth]{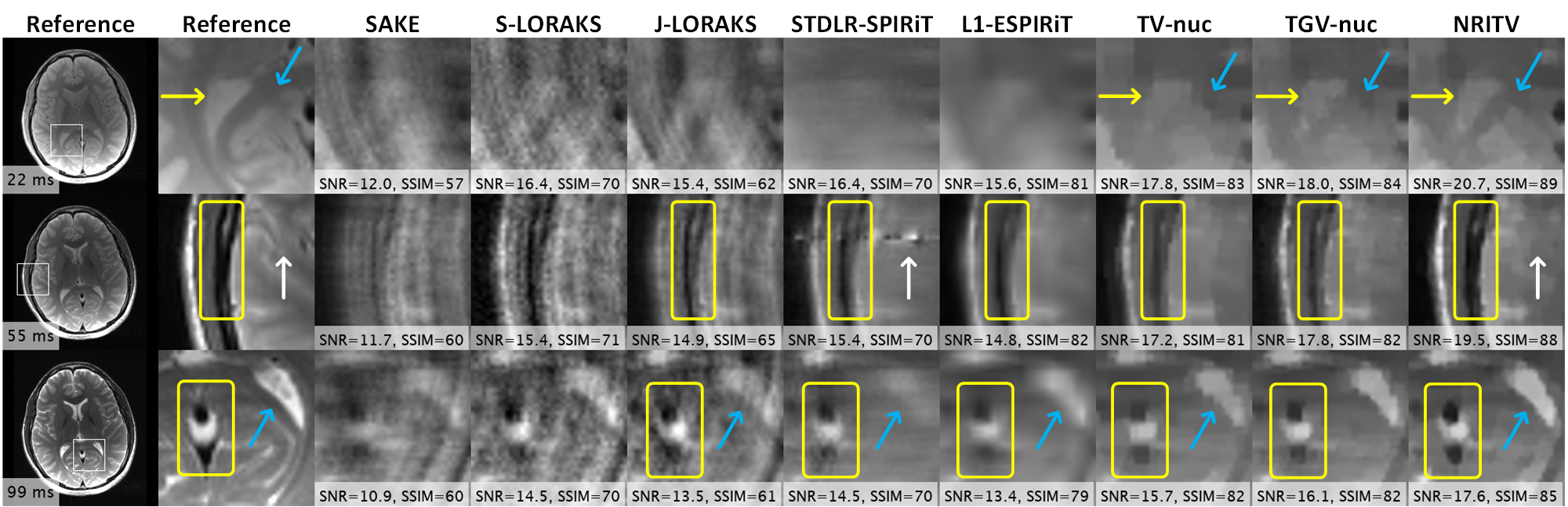}
    \caption{Magnified views of the results for the brain data at $R=7$ and $\sigma=0.02$. Some errors in the competing methods and their corrections with the proposed one are marked. {Appendix Fig. \ref{Supp_Itthi} shows full size FOV's, error maps and the sampling pattern for this experiment.}}
    \label{showpice_Itthi} \end{figure*} 

Fig. \ref{showpice_Itthi} shows reconstructions for a triple-contrast in-vivo brain dataset from  \cite{Itthi} graciously provided by Dr. Chatnuntawech. The data had been collected from a healthy volunteer with Institutional Review Board
approval and informed consent at 3T (MAGNETOM Trio, A Tim
System, Siemens, Erlangen, Germany) using a turbo spin-echo
(TSE) sequence with the important parameters: FOV = $22 \times 22 \text{ cm}^2$, voxel size = $0.9\times0.9\times3 \text{ mm}^3$,
 repetition time (TR) = $4$ s and echo times (TE) = $22,55,99$ ms with 32 receive coils. However, for speed, 
 ESPIRiT's coil compression was used to reduce this number to 8. 
 White Gaussian noise with standard deviation $\sigma=0.02$ was added, then k-space was retrospectively undersampled at $ R=7$. Fig. \ref{showpice_Itthi} shows that SAKE and the LORAKS-based solutions have failed to denoise and de-alias the data. STDLR-SPIRiT, L1-ESPIRiT, TV-nuc and TGV-nuc output over-smoothed solutions (e.g. blue arrows in the $22$- and $99$-ms contrasts pointing to the temporal horns and white matter, yellow boxes encompassing the straight sinus in the $22$- and $99$-ms contrasts). Loss of contrast at skull is observed with the same methods (yellow boxes in the $55$-ms contrast). NRITV resolves all these drawbacks and finds the best balance between noise removal and detail retention. {Full-size views for all methods along with error maps and sampling pattern can be observed in Fig. \ref{Supp_Itthi}.  Selected parameters for the competing methods are given in Table \ref{Table_Ithhi}. NRITV used $\lambda =3\times 10^{-3}$ and ESPIRiT-based maps with $22\%$ error.}
 
 \begin{figure}\centering \includegraphics[width=0.6\linewidth]{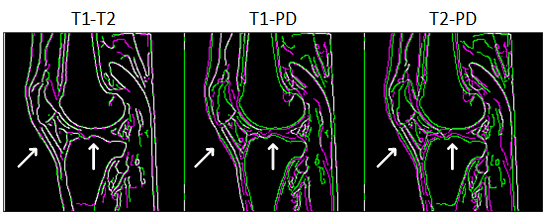}\caption{Edge detection for the triple-contrast knee scan. Each image shows the superimpose of the edges for 2 contrasts. White curves show aligned edges while disjoint green and blue curves indicate mismatch. $T_1$ and $T_2$ contrasts are matched whereas the PD contrast is misaligned with them.} \label{mismatch} \end{figure}

  A triple-contrast knee dataset from \cite{JGBRWT} is used in our next experiment. The data was acquired from GE 3 Tesla
 scanner (Discovery MR750W, USA) with important parameters
  TR = $499$ ms and TE = $9.63$ ms for TSE $T_1$-weighted contrast; TR =
 $2435$ ms and TE = $49.98$ ms for $T_2$-weighted contrast and
 TR = $2253$ ms and TE = $31.81$ ms for PD-weighted contrast with FOV = $180 \times 180 \text{ mm}^2$  and
 slice thickness = $4$ mm. 
 As is the way with many MC free-breathing sequences, some degree of involuntary patient movement is anticipated in this dataset. Canny's edge detection algorithm \cite{Canny} was used to identify the location of the edges in all three contrasts. In each of the three panels of Fig. \ref{mismatch}, the results of edge detection for each of the three possible pairs of images is shown. Disjoint green and blue curves represent mismatched edges, whereas matched edges are colored white. It is clear that $T_1$ and $T_2$ contrasts almost perfectly match but the PD contrast is misaligned with them. The arrows in Fig. \ref{mismatch} mark example regions where the PD-weighted contrast is misaligned with the other two contrasts. 
This dataset is used to see if the proposed method can sustain some degree of misalignment. To further complicate the situation, white Gaussian noise with $\sigma=0.02$ was added and k-space was undersampled at $R=5$. The results are magnified in Fig. \ref{showpice_Knee}. SAKE and LORAKS-based methods are again dominated by noise while other compared methods lose contrast (yellow boxes encompassing the patellar tendon in the PD results) and fine details (blue arrows marking the border between tibia and infrapattelar fat in $T_1$ and $T_2$ contrasts). NRITV effectively removes noise without causing any loss of such details. {Full-size FOVs, error maps and the sampling pattern related to this experiment are shown in Supplementary Fig. \ref{Supp_knee}. NRITV used $\lambda=5\times 10^{-3}$ while the parameters of the other methods are presented in Table \ref{Table_knee}. ESPIRiT estimated the maps used in NRITV with $14\%$ RLNE.}

 \begin{figure} 
	\includegraphics[width = \textwidth]{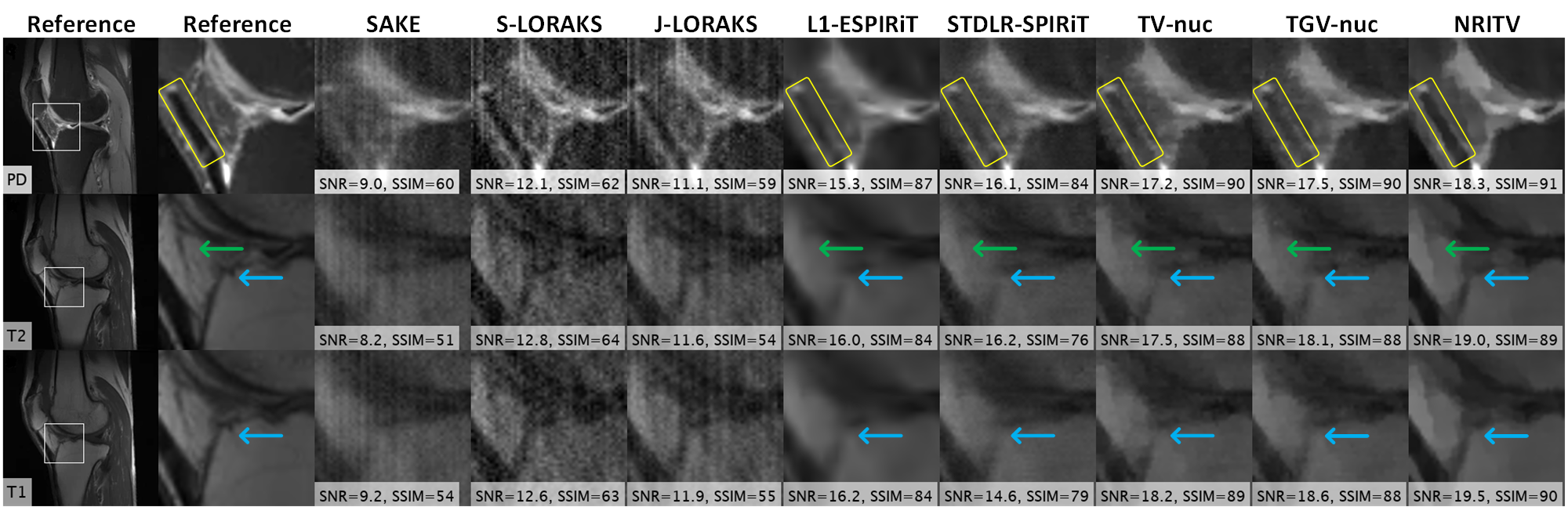}
	\caption{Magnified views of the results for the knee data at $R=5$ and $\sigma = 0.02$. Some errors in the competing methods and their corrections with the proposed one are marked by arrows and boxes. {Appendix  Fig. \ref{Supp_knee} shows full size FOVs, error maps and the sampling pattern for this experiment.}}
	\label{showpice_Knee} \end{figure}
        
 \begin{figure*}           \includegraphics[width = \textwidth]{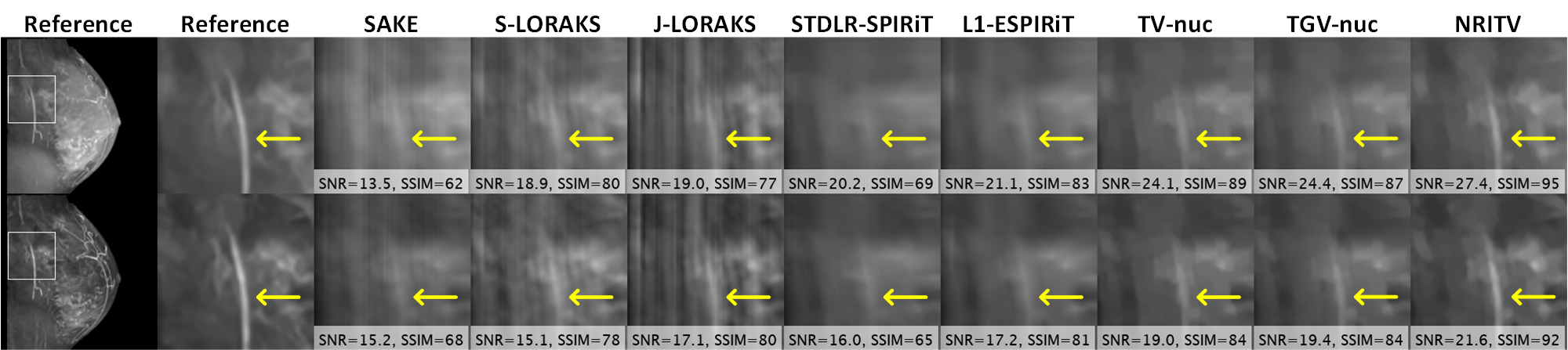}          \caption{Magnified views on 2 contrasts of the 4-contrast breast data at $R=9$. An example error in the competing methods and its correction in the proposed one is marked with arrows. {Appendix  Fig. \ref{Supp_breast} shows full size FOVs for all 4 contrasts, error maps and the sampling pattern in this experiment.}}          \label{showpice2} \end{figure*} 

The next experiment involves reconstruction of a 4-contrast breast dataset from Cancer Imaging Archive \cite{CIA} which was acquired 
from a 33 year old patient on a 1.5T MRI machine (GE Healthcare, Chicago, Illinois, U.S.). Sequence parameters can be seen in \cite{CIA}. 
The 5-coil data was retrospectively undersampled at $R=9$ and reconstructed by all methods.  Fig. \ref{showpice2} shows magnified views on the first 2 contrasts. The arrows in Fig. \ref{showpice2} mark a contrast-enhanced vessel which is highly aliased by LORAKS methods and poorly captured by the other compared works but sharply visible with the proposed one. { Fig. \ref{Supp_breast} shows full FOVs for all 4 contrasts reconstructed by all methods. Table \ref{Table_breast} gives the parameters for compared methods. In this experiment, ESPIRiT failed to produce an acceptable estimation of sensitivity maps, hence NRITV used the same maps as the SENSE-type methods.}

  
\section{Discussion and Conclusions}\label{sec concl}

\begin{figure}\centering \includegraphics[width=0.55\linewidth]{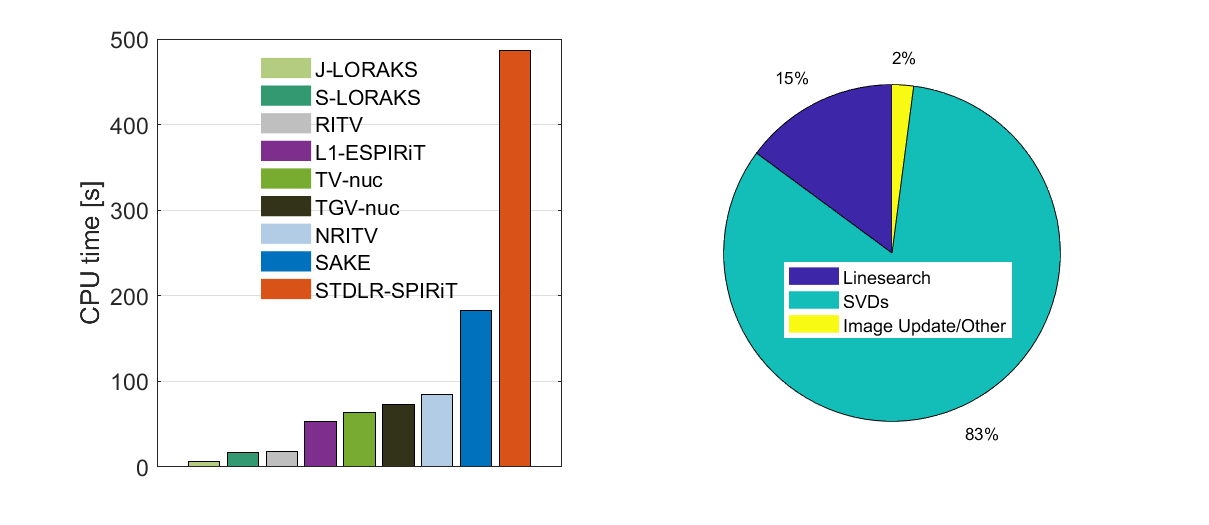}           \caption{Average CPU times required by different methods at size $n=80$ (left) and CPU time analysis for the proposed algorithm at $n=200$ (right). STDLR-SPIRiT is by far the most expensive method despite bypassing SVDs. On the other hand, \textcolor{black}{most} of the computational effort in the proposed algorithm is taken by SVDs, which can be optimized by parallel computing and C/C++ re-implementation of SVDs.  } \label{times} \end{figure}
This article introduced a novel isotropic multi-{contrast} image regularizer within the context of compressed multi-contrast parallel MRI and a primal-dual algorithm with linesearch was proposed to solve the new model. It was observed that this method significantly improves upon state-of-the-art approaches by virtue of inherent isotropy, aligned edges and leakage-prevention of {contrast}-specific details. Robustness of the proposed method was validated with in-vivo data where imperfections such as inaccurate estimations of coil sensitivity maps, high levels of noise, aggressive undersampling or cross-contrast mismatch due to patient motion did not impact the performance of the proposed strategy, while other state-of-the-art approaches broke down.       
The proposed method may carry over to other multi-{channel} image tasks such as color image processing, as well as other multi-modality medical imaging applications such as PET-CT, MR-PET, PET-SPECT etc. 
The proposed method might as well have been called NRITV-ESPIRiT following the traditions in the literature. However, such a title would undermine the fact that NRITV is not reliant on ESPIRiT and any acceptable estimation of coil maps could be used in NRITV.
Currently, the proof-of-principles implementation of the proposed method takes about $6$ seconds to run per iteration on 4 contrasts at size $n=200$. This can improve by an optimized, parallelized C++ transfer of the codes, particularly the SVDs which take most of the computational time 
(see Fig. \ref{times}).

	\begin{table} \centering	
		\caption{Memory cost for variational methods}
	\color{black}	\begin{tabular}{c c c c}
	\toprule
			method & primal variables & dual variables & total memory cost\\
			\midrule NRITV & $u_{n\times n}, (v_s)_{n\times n \times 4}$ & $r_{n\times n}, h_{n\times n \times 2}$ & $8$ matrices of size $n\times n$	\\
			TGV \cite{TGV MRI} & $u_{n\times n}, v_{n\times n \times 2}$ & $r_{n\times n}, p_{n\times n \times 2}, q_{n\times n \times 3}$ & $9$ matrices of size $n\times n$	\\
			TV \cite{TGV MRI} & $u_{n\times n}$ & $r_{n\times n}, p_{n\times n \times 2}$ & $4$ matrices of size $n\times n$	\\ \bottomrule\label{TableMemoryCost}\end{tabular} \end{table}  
\textcolor{black}{	Moreover, since NRITV is in the same family of variational regularizers as TV and TGV, it might be interesting to compare them in terms of memory cost. Table \ref{TableMemoryCost} presents the memory cost for the three variational methods, TV, TGV and NRITV. Comparing Table \ref{TableMemoryCost} with Fig. \ref{times} reveals the interesting fact that while TGV has a higher memory cost than NRITV, the former runs faster than the latter. This is because the dual variables in NRITV are placed in a linesearch loop which inevitably increases run time. Note that the values in table \ref{TableMemoryCost} are given for $N=P=1$. For the general case, the total memory cost will have to be trivially multiplied by $N\times P$.}



%





\section*{Declarations}
\begin{itemize}
	\item Funding: This work is an independent research which was not financially supported.
	\item Conflicts of interest/Competing interests: Not applicable.
    \item Availability of data and code: The data and software reproducing the proposed method are publicly available at \cite{Mend}.
	
\end{itemize}
\newpage
\appendix \section*{Appendix: more details on the experiments}  \label{app}

\begin{figure}[H] \includegraphics[width=\textwidth]{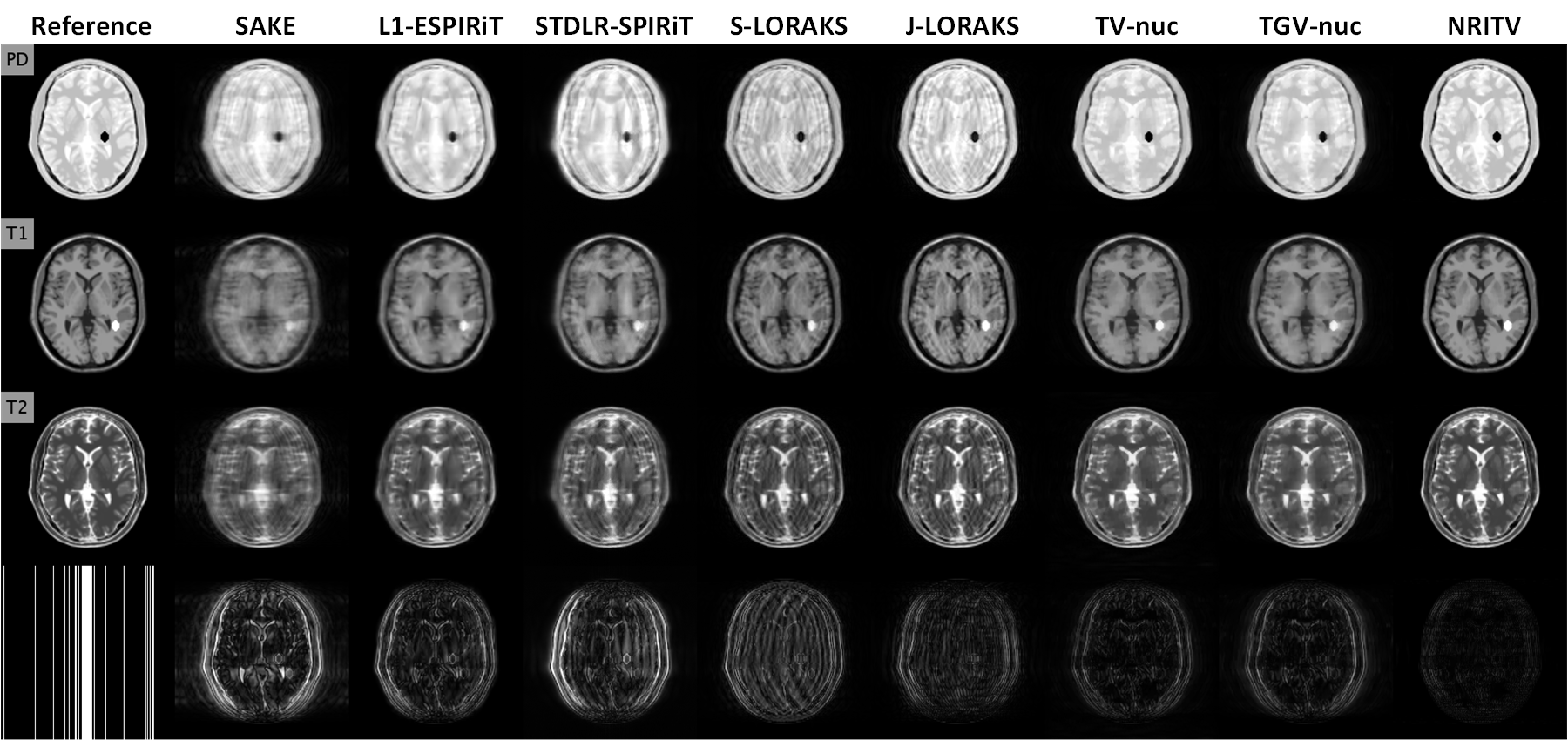} \caption{Full FOVs of reconstructions in Fig. \ref{showpice} in the main article. The bottom row shows the sampling pattern used in the experiment as well as error maps (sum of absolute differences over all contrasts) for each method.} \label{Supp_Lesion}\end{figure}

\begin{figure}[H] \includegraphics[width=\textwidth]{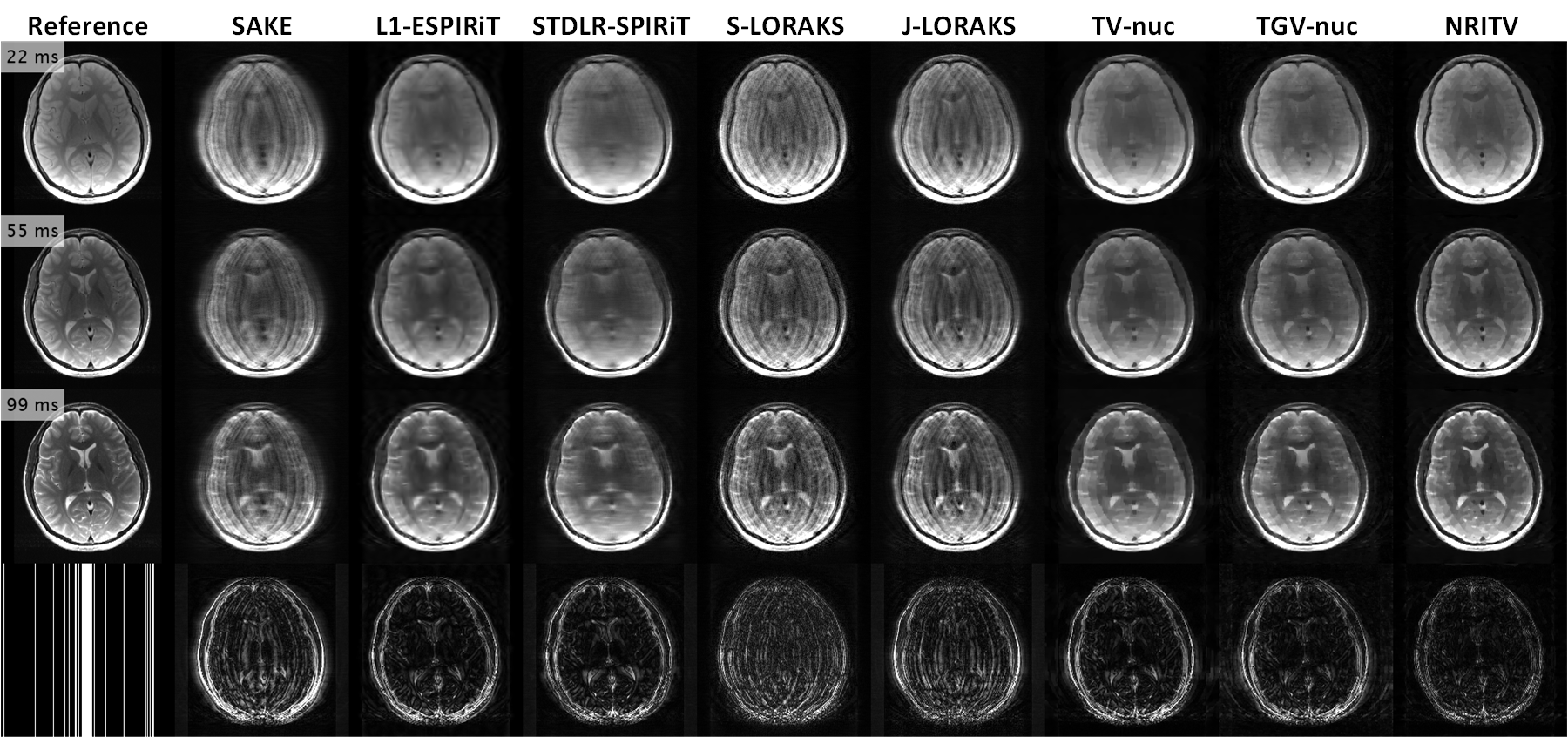}           \caption{Full FOVs of reconstructions in Fig. \ref{showpice_Itthi} in the main article. The bottom row shows the sampling pattern used in the experiment as well as error maps (sum of absolute differences over all contrasts) for each method.}\label{Supp_Itthi} \end{figure}

\begin{figure} \includegraphics[width=\textwidth]{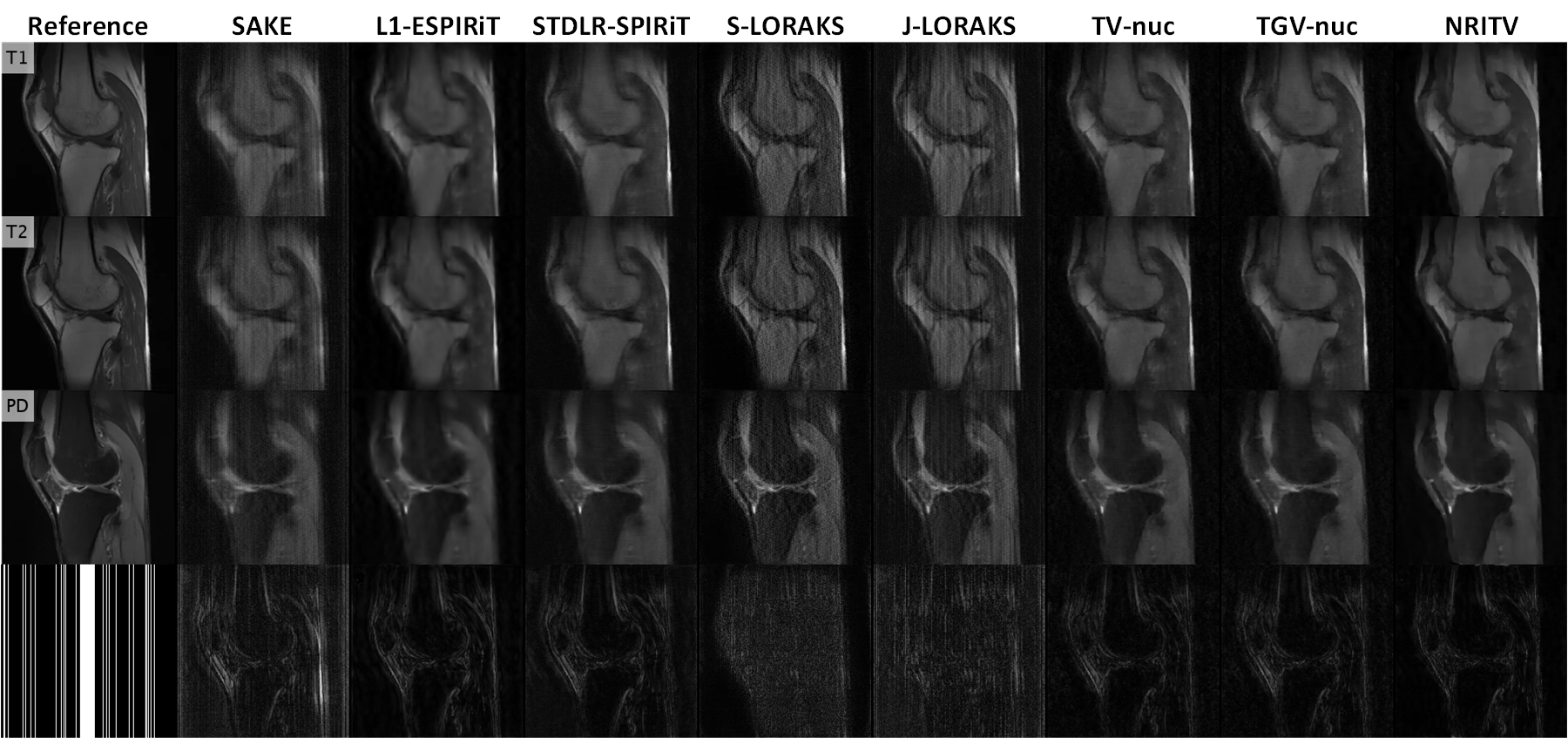}           \caption{Full FOVs of reconstructions in Fig. \ref{showpice_Knee} in the main article. The bottom row shows the sampling pattern used in the experiment as well as error maps (sum of absolute differences over all contrasts) for each method.}\label{Supp_knee} \end{figure} 

\begin{figure} \includegraphics[width=\textwidth]{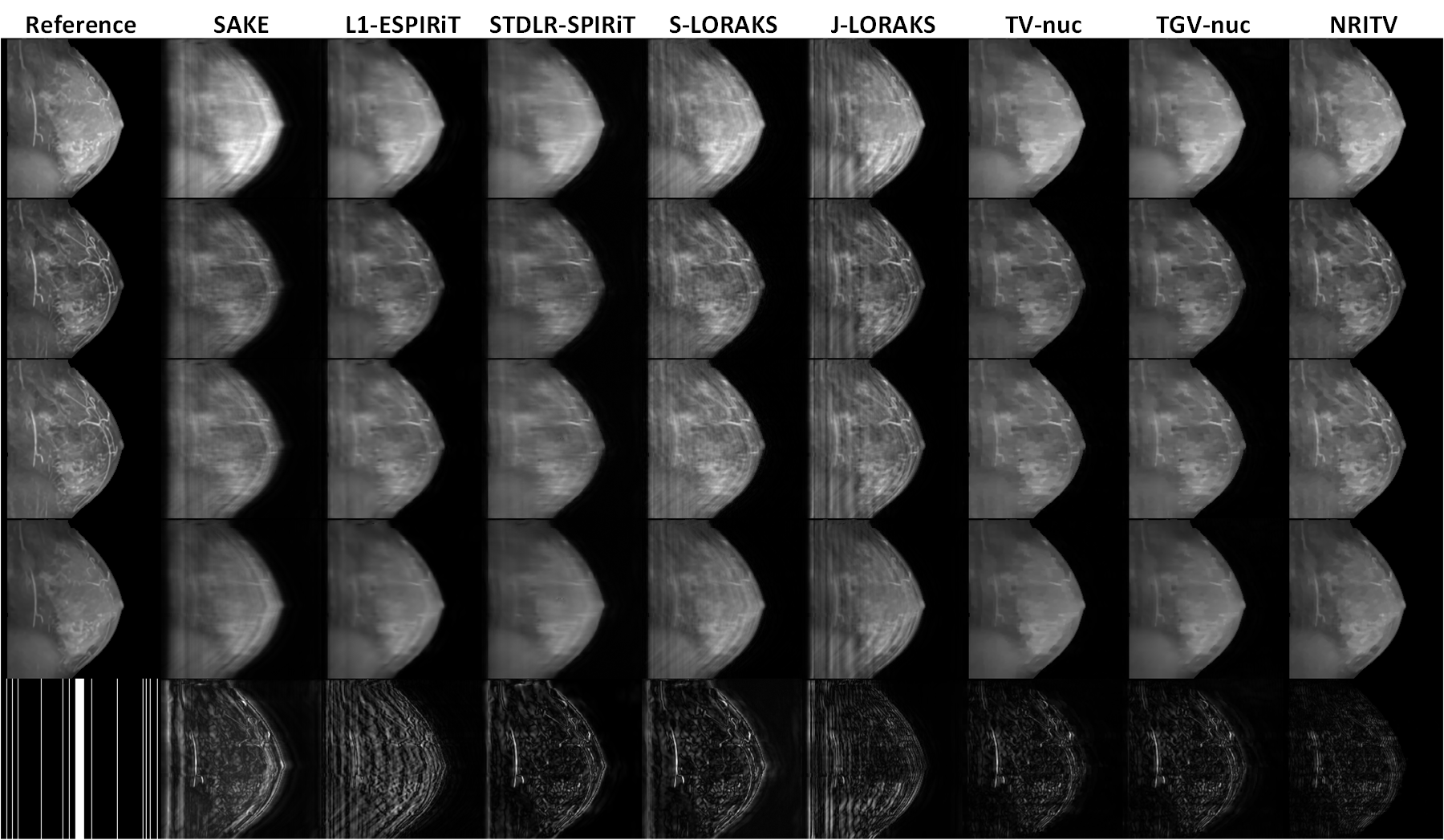}           \caption{Full FOVs of reconstructions in Fig. \ref{showpice2} in the main article. The bottom row shows the sampling pattern used in the experiment as well as error maps (sum of absolute differences over all contrasts) for each method.}\label{Supp_breast} \end{figure}

\begin{figure} \includegraphics[width=\textwidth]{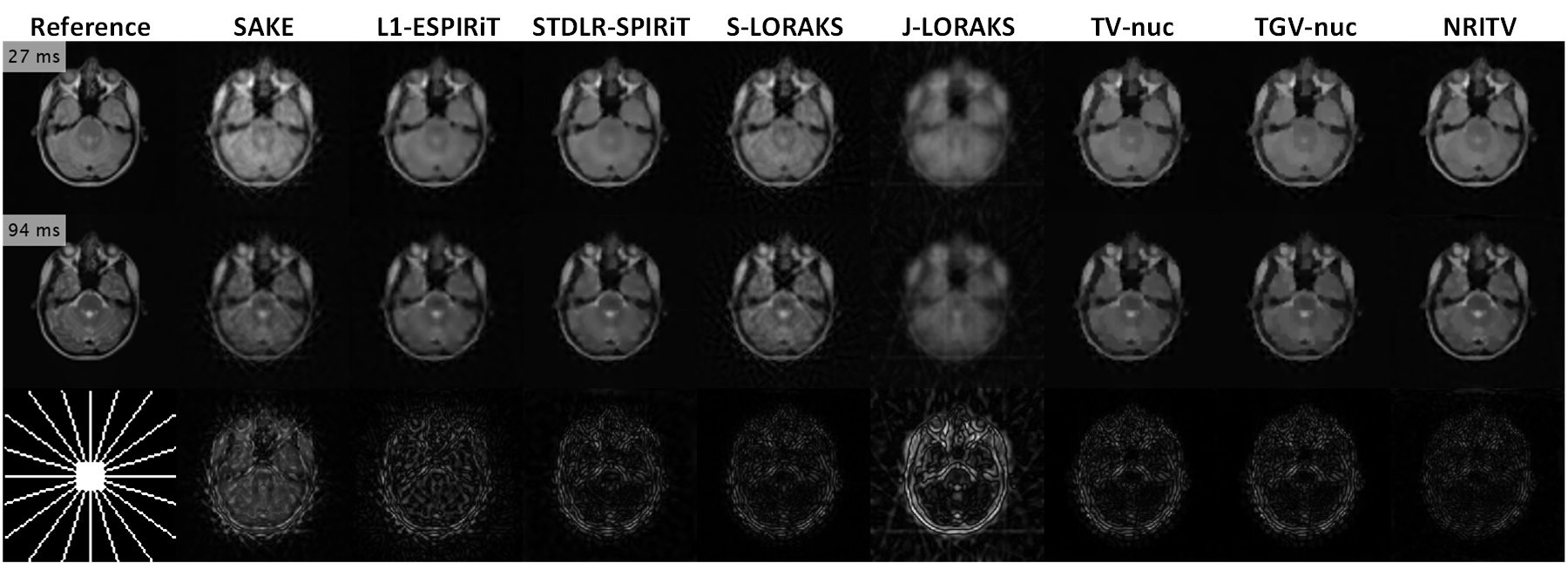}           \caption{Reconstruction results for a radially sampled ($R=7$) double-echo ($27$ and $94$ ms) slice from the in-vivo brain data in reference [11]. It can be seen in the error maps that most of the methods output acceptable results with NRITV being superior. As usual, SENSE-type methods used perfect maps and NRITV used ESPIRiT-estimated maps ($18\%$ error). Compared methods' parameters were mostly the same as the breast experiment.} \label{Supp_radial}\end{figure}


 \newpage
\begin{table}   \centering
	\caption{Parameter Selections in Fig. \ref{BW_bars} ($R=5$)}
	\scalebox{0.85}{  
		\begin{tabular}{ccccccccccc}    \toprule  
			methods & kernel & calib. param. & rank & radius & $\lambda$ & $\alpha_0$ & $\alpha_1$ & $\mu$ &  iters. & $\epsilon$\\ \midrule
			SAKE & $4 \times 4$ & - & $2$ & - & - & - & - & - & $500$ &- \\
			L1-ESPIRiT & $4 \times 4$ & $2\times 10^{-2}$ & - & - & $7\times 10^{-3}$ & - & - & - & $100$ &- \\
			STDLR-SPIRiT & $3\times 3$ & $5\times 10^{-5}$ & - & - & $(10^4,10^6)$ & - & - & - & $15$ & - \\
			S-LORAKS & - & - & $40$ & $3$ & $5\times 10^{-3}$ & - & - & - & $50$ & $10^{-3}$ \\
			J-LORAKS & - & - & $80$ & $1$ & - & - & - & - & $50$ &$10^{-6}$ \\
			TV-nuc & - & - & - & - & $2 \times 10^{-4}$ & - & - & $5\times 10^{-3}$ & $500$& $10^{-5}$  \\
			TGV-nuc & - & - & - & - & $5$ & $3.5 \times 10^{-3}$ & $1.8 \times 10^{-3}$ & $5\times10^{-2}$ & $500$ & $10^{-5}$\\
			\bottomrule     
		\end{tabular}
	} \label{Table_BWR5}
\end{table}

\begin{table}   \centering
	\caption{Parameter Selections in Fig. \ref{BW_bars} ($R=7$)}  
	\scalebox{0.85}{
		\begin{tabular}{lcccccccccc}    \toprule  
			methods & kernel & calib. param. & rank & radius & $\lambda$ & $\alpha_0$ & $\alpha_1$ & $\mu$ &  iters & $\epsilon$\\ \midrule
			SAKE & $4 \times 4$ & - & $2$ & - & - & - & - & - & $500$ & -  \\
			L1-ESPIRiT & $3 \times 3$ & $2\times 10^{-2}$ & - & - & $2.5\times 10^{-3}$ & - & - & - & $100$ & - \\
			STDLR-SPIRiT & $3\times 3$ & $5\times 10^{-5}$ & - & - & $(10^4,10^6)$ & - & - & - & $15$ & - \\
			S-LORAKS & - & - & $40$ & $3$ & $5\times 10^{-3}$ & - & - & - & $50$ & $10^{-3}$ \\
			J-LORAKS & - & - & $94$ & $1$ & - & - & - & - & $150$& $10^{-6}$  \\
			TV-nuc & - & - & - & - & $2 \times 10^{-4}$ & - & - & $5\times 10^{-3}$ & $500$ & $10^{-5}$ \\
			TGV-nuc & - & - & - & - & $5$ & $3.5 \times 10^{-3}$ & $1.8 \times 10^{-3}$ & $5\times 10^{-2}$ & $500$& $10^{-5}$  \\
			\bottomrule     
		\end{tabular}
	}\label{Table_BWR7}
\end{table}

\begin{table}   \centering
	\caption{Parameter Selections in Fig. \ref{showpice} and Fig. \ref{Supp_Lesion} (phantom lesion). \textcolor{black}{Optimization plot for the boldfaced parameter is given in Fig. \ref{param_opt} (c).}  }
	\scalebox{0.85}{
		\begin{tabular}{lcccccccccc}    \toprule  
			methods & kernel & calib. param. & rank & radius & $\lambda$ & $\alpha_0$ & $\alpha_1$ & $\mu$ &  iters & $\epsilon$\\ \midrule
			SAKE & $5 \times 5$ & - & $2$ & - & - & - & - & - & $100$ &-  \\
			L1-ESPIRiT & $6 \times 6$ & $2\times 10^{-2}$ & - & - & $2.5\times 10^{-3}$ & - & - & - & $100$& -  \\
			STDLR-SPIRiT & $5\times 5$ & $2\times 10^{-2}$ & - & - & $(10^4,10^6)$ & - & - & - & $15$ & - \\
			S-LORAKS & - & - & $40$ & $3$ & $5\times 10^{-3}$ & - & - & - & $50$ & $10^{-3}$  \\
			J-LORAKS & - & - & $\mathbf{160}$ & $2$ & - & - & - & - & $100$ & $10^{-6}$ \\
			TV-nuc & - & - & - & - & $2 \times 10^{-4}$ & - & - & $5\times 10^{-3}$ & $500$& $10^{-5}$  \\
			TGV-nuc & - & - & - & - & $5$ & $3.5\times 10^{-3}$ & $1.8 \times 10^{-3}$ & $10^{-1}$ & $500$& $10^{-5}$  \\
			\bottomrule     
		\end{tabular}
	}\label{Table_phantom}
\end{table}
\begin{table}   \centering
	\caption{Parameter Selections in Fig. \ref{showpice_Itthi} and Fig. \ref{Supp_Itthi} (brain). \textcolor{black}{Optimization plot for the boldfaced parameter is given in Fig. \ref{param_opt} (a).}} 
	\scalebox{0.85}{ 
		\begin{tabular}{lcccccccccc}    \toprule  
			methods & kernel & calib. param. & rank & radius & $\lambda$ & $\alpha_0$ & $\alpha_1$ & $\mu$ &  iters & $\epsilon$\\ \midrule
			SAKE & $4 \times 4$ & - & $1$ & - & - & - & - & - & $100$ &-  \\
			L1-ESPIRiT & $5 \times 5$ & $ 10^{-2}$ & - & - & $ 10^{-2}$ & - & - & - & $100$ &-  \\
			STDLR-SPIRiT & $5\times 5$ & $ 10^{-4}$ & - & - & $(10^3,10^6)$ & - & - & - & $25$ & - \\
			S-LORAKS & - & - & $40$ & $3$ & $10^{-2}$ & - & - & - & $50$& $10^{-3}$  \\
			J-LORAKS & - & - & $250$ & $3$ & - & - & - & - & $50$& $10^{-6}$ \\
			TV-nuc & - & - & - & - &  $\mathbf{2 \times 10^{-3}}$ & - & - & $5\times 10^{-3}$ & $500$ & $10^{-5}$ \\
			TGV-nuc & - & - & - & - & $1$ & $10^{-2}$ & $5 \times 10^{-3}$ & $10^{-2}$ & $500$ & $10^{-5}$ \\
			\bottomrule     
		\end{tabular}
	}\label{Table_Ithhi}
\end{table}   
\begin{table}   \centering
	\caption{Parameter Selections in Fig. \ref{showpice_Knee} and Fig. \ref{Supp_knee} (knee). \textcolor{black}{Optimization plots for the boldfaced parameters are given in Fig. \ref{param_opt} (d) and (b).}}
	\scalebox{0.85}{  
		\begin{tabular}{lcccccccccc}  
			
			\toprule  
			methods & kernel & calib. param. & rank & radius & $\lambda$ & $\alpha_0$ & $\alpha_1$ & $\mu$ &  iters. & $\epsilon$\\ \midrule
			SAKE & $5 \times 5$ & - & $3$ & - & - & - & - & - & $100$&-  \\
			L1-ESPIRiT & $6 \times 6$ & $\mathbf{2\times 10^{-2}}$ & - & - & $ 10^{-2}$ & - & - & - & $100$ &-  \\
			STDLR-SPIRiT & $5\times 5$ & $ 10^{-4}$ & - & - & $(10^4,10^7)$ & - & - & - & $15$ & - \\
			S-LORAKS & - & - & $50$ & $3$ & - & - & - & - & $50$ & $10^{-3}$ \\
			J-LORAKS & - & - & $450$ & $2$ & - & - & - & - & $50$ & $10^{-6}$ \\
			TV-nuc & - & - & - & - & $6 \times 10^{-3}$ & - & - & $5\times 10^{-3}$ & $500$& $10^{-5}$  \\
			TGV-nuc & - & - & - & - & $1$ & $\mathbf{ 10^{-2}}$ & $5 \times 10^{-3}$ & $10^{-2}$ & $500$& $10^{-5}$  \\
			\bottomrule     
		\end{tabular}
	}\label{Table_knee}
\end{table}         

\begin{table}   \centering
	\caption{Parameter Selections in Fig. \ref{showpice2} and Fig. \ref{Supp_breast} (breast). \textcolor{black}{Optimization plots for the boldfaced parameters are given in Fig. \ref{param_opt} (f) and (e).}}  
	\scalebox{0.85}{
		\begin{tabular}{lcccccccccc}    \toprule  
			methods & kernel & calib. param. & rank & radius & $\lambda$ & $\alpha_0$ & $\alpha_1$ & $\mu$ &  iters. & $\epsilon$\\ \midrule
			SAKE & $3 \times 3$ & - & $2$ & - & - & - & - & - & $100$&-  \\
			L1-ESPIRiT & $3 \times 3$ & $2\times 10^{-2}$ & - & - & $ 2.5 \times 10^{-3}$ & - & - & - & $100$ &-  \\
			STDLR-SPIRiT & $\mathbf{5\times 5}$ & $ 10^{-4}$ & - & - & $(10^3,10^7)$ & - & - & - & $15$ & - \\
			S-LORAKS & - & - & $50$ & $\mathbf{3}$ & $10^{-3}$ & - & - & - & $50$ & $10^{-3}$ \\
			J-LORAKS & - & - & $150$ & $2$ & - & - & - & - & $150$ & $10^{-6}$ \\
			TV-nuc & - & - & - & - & $2 \times 10^{-4}$ & - & - & $5\times 10^{-3}$ & $500$& $10^{-5}$  \\
			TGV-nuc & - & - & - & - & $5$ & $ 3.5\times 10^{-3}$ & $1.8 \times 10^{-3}$ & $10^{-1}$ & $500$& $10^{-5}$  \\
			\bottomrule     
		\end{tabular}
	}\label{Table_breast}
\end{table}         

\begin{figure}
\includegraphics[width=\textwidth]{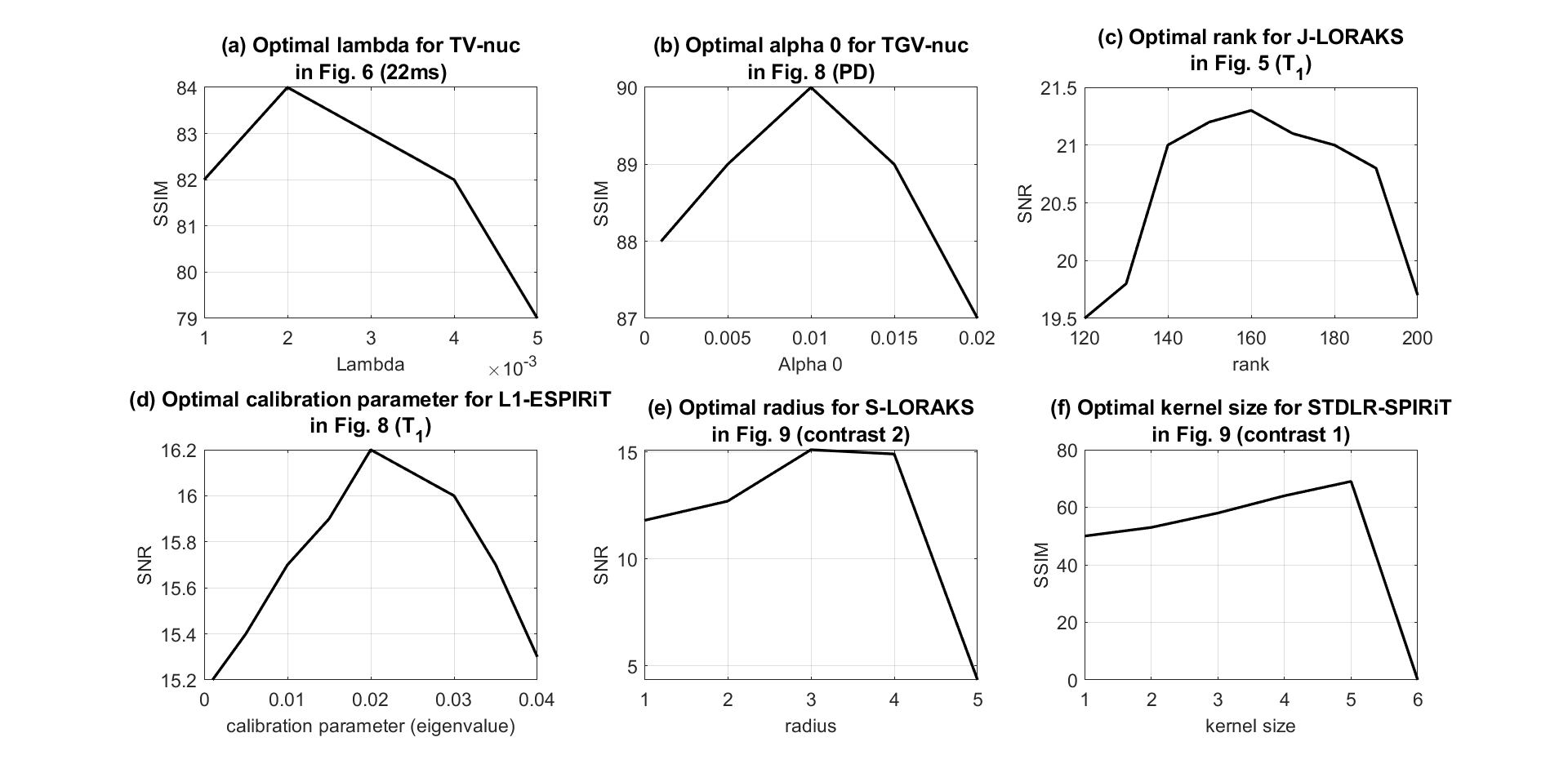}
\caption{Optimization process for some of the parameters of the compared works.}\label{param_opt}
\end{figure}

\vspace{4cm}
\textbf{Erfan Ebrahim Esfahani} received his BSc in mathematics from Amirkabir University of Technology (Tehran Polytechnic) in 2016 and MSc in applied mathematics from the University of Tehran in 2018. During his time in the school of mathematics, his study mainly concerned signal and image analysis. He is currently an independent researcher and searching for available PhD programs in biomedical fields focused on his research interests including radiology, medical imaging, nuclear medicine and image-guided interventions.  

\end{document}